\documentclass[a4paper,USenglish]{lipics-v2019} 

\usepackage{amsthm}
\usepackage{amsmath}

\theoremstyle{plain}

\newtheorem{assumption}[theorem]{Assumption}

\def\hmath$#1${\texorpdfstring{{\rmfamily{#1}}}{#1}}

\usepackage{algorithm}
\usepackage{algorithmic}
\usepackage{amsfonts}

\usepackage{url}

\hideLIPIcs

\title{Who started this rumor? Quantifying the natural differential privacy of gossip protocols} 

\titlerunning{Quantifying the Natural Differential Privacy Guarantees of Gossip Protocols} 

\author{Aur\'elien Bellet}{INRIA}{aurelien.bellet@inria.fr}{https://orcid.org/0000-0003-3440-1251}{This research was supported by grants ANR-16-CE23-0016-01 and ANR-18-CE23-0018-03, by the European Union's Horizon 2020 Research and Innovation Program under Grant Agreement No. 825081 COMPRISE (\url{https://project.inria.fr/comprise/}), by a grant from CPER Nord-Pas de Calais/FEDER DATA Advanced data science and technologies 2015-2020.}

\author{Rachid Guerraoui}{EPFL}{rachid.guerraoui@epfl.ch}{https://orcid.org/0000-0002-4794-8902}{This research was supported by European ERC Grant 339539 - AOC.}

\author{Hadrien Hendrikx}{PSL, DIENS, INRIA}{hadrien.hendrikx@inria.fr}{}{This research was supported by the MSR-INRIA joint centre.}

\authorrunning{A. Bellet, R. Guerraoui and H. Hendrikx} 

\Copyright{Aur\'elien Bellet, Rachid Guerraoui and Hadrien Hendrikx} 



\category{} 

\relatedversion{} 

\supplement{}

\nolinenumbers

\EventEditors{Hagit Attiya}
\EventNoEds{1}
\EventLongTitle{34rd International Symposium on Distributed Computing (DISC 2020)}
\EventShortTitle{DISC 2020}
\EventAcronym{DISC}
\EventYear{2020}
\EventDate{October 12--18, 2020}
\EventLocation{Virtual Conference}
\EventLogo{}
\SeriesVolume{179}
\ArticleNo{3}

\begin{document}

\maketitle

\begin{abstract}
Gossip protocols (also called rumor spreading or epidemic protocols) are widely used to disseminate information in massive peer-to-peer networks. These protocols are often claimed to guarantee privacy because of the  uncertainty they introduce on the node that started the dissemination. But is that claim really true? Can the source of a gossip safely hide in the crowd? This paper examines, for the first time, gossip protocols through a rigorous mathematical framework based on differential privacy to determine the extent to which the source of a gossip can be traceable. Considering the case of a complete graph in which a subset of the nodes are curious, we study a family of gossip protocols parameterized by a ``muting'' parameter $s$: nodes stop emitting after each communication with a fixed probability $1-s$. We first prove that the standard push protocol, corresponding to the case $s=1$, does not satisfy differential privacy for large graphs. In contrast, the protocol with $s=0$ (nodes forward only once) achieves optimal privacy guarantees but at the cost of a drastic increase in the spreading time compared to standard push, revealing an interesting tension between privacy and spreading time. Yet, surprisingly, we show that some choices of the muting parameter $s$ lead to protocols that achieve an optimal order of magnitude in both privacy and speed. Privacy guarantees are obtained by showing that only a small fraction of the possible observations by curious nodes have different probabilities when two different nodes start the gossip, since the source node rapidly stops emitting when $s$ is small. The speed is established by analyzing the mean dynamics of the protocol, and leveraging concentration inequalities to bound the deviations from this mean behavior. We also confirm empirically that, with appropriate choices of $s$, we indeed obtain protocols that are very robust against concrete source location attacks (such as maximum a posteriori estimates) while spreading the information almost as fast as the standard (and non-private) push protocol.
\end{abstract}

\section{Introduction}



\emph{Gossip} protocols (also called \emph{rumor spreading} or \emph{epidemic protocols}), in which  participants \emph{randomly} choose a neighbor to communicate with, are both simple and efficient means to exchange information in P2P networks \cite{frieze1985shortest,pittel1987spreading,karp2000randomized,berenbrink2010efficient}.
They are a basic building block to propagate and aggregate information in distributed databases \cite{demers1987epidemic,boyd2006randomized}
and social networks \cite{doerr2011social,giakkoupis2015privacy}, to model the spread of infectious diseases \cite{SIS},  as well as to train machine learning models on distributed datasets \cite{Duchi2012a,colin2016gossip,vanhaesebrouck2017decentralized,Koloskova2019}.

Some of the information gossiped may be sensitive, and participants sharing it may not want to be identified. This can for instance be the case of whistle-blowers or individuals that would like to exercise their right to freedom of expression in totalitarian regimes. Conversely, it may sometimes be important to locate the source of a (computer or biological) virus, or fake news, in order to prevent it from spreading before too many participants get ``infected''. 

There is a folklore belief that gossip protocols inherently guarantee some form of \emph{source anonymity} because participants cannot know who issued the information in the first place \cite{ghaffari2016discreetly}. Similarly, identifying ``patient zero'' for real-world epidemics is known to be a very hard task. Intuitively indeed, random and local exchanges make identification harder. But to what extent?
Although some work has been devoted to the design of source location strategies in specific settings \cite{jiang2017identifying,pinto2012locating,shah2011rumors}, the general anonymity claim has never been studied from a pure \emph{privacy} perspective, that is, independently of the very choice of a source location technique.
Depending on the use-case, it may be desirable to have strong privacy guarantees (e.g., in anonymous information dissemination) or, on the contrary, we may hope for weak guarantees, e.g., when trying to identify the source of an epidemic.
In both cases, it is crucial to precisely quantify the anonymity level of gossip protocols and study its theoretical limits through a principled approach. This is the challenge we take up in this paper for the classic case of gossip dissemination in a complete network graph.

Our first contribution is an information-theoretic model of anonymity in gossip protocols based on \emph{$(\epsilon,\delta)$-differential privacy} (DP) \cite{dwork2011differential,dwork2006our}. Originally introduced in the database community, DP is a precise mathematical framework  recognized as the gold standard for studying the privacy guarantees of information release protocols. In our proposed model, the information to protect is the source of the gossip, and an adversary tries to locate the source by monitoring the communications (and their relative order) received by a subset of $f$ curious nodes. In a computer network, these curious nodes may have been compromised by a surveillance agency; in our biological example, they could correspond to health professionals who are able to identify whether a given person is infected.
Our notion of DP then requires that the probability of any possible observation of the curious nodes is almost the same no matter who is the source, thereby limiting the predictive power of the adversary regardless of its actual source location strategy. A distinctive aspect of our model is that the mechanism that seeks to ensure DP comes only from the \emph{natural} randomness and partial observability of gossip protocols, not from additional perturbation or noise which affects the desired output as generally needed to guarantee DP \cite{Dwork2014a}.
We believe our adaptation of DP to the gossip context to be of independent interest. We also complement it with a notion of \emph{prediction uncertainty} which guarantees that even unlikely events do not fully reveal the identity of the source under a uniform prior on the source. This property directly upper bounds the probability of success of any source location attack, including the maximum likelihood estimate.

We use our proposed model to study the privacy guarantees of a generic family of gossip protocols parameterized by a \emph{muting parameter} $s$: nodes have a fixed probability $1-s$ to stop emitting after each communication (until they receive the rumor again). In our biological parallel, this corresponds to the fact that a person stops infecting other people after some time. The muting parameter captures the ability of the protocol to forget initial conditions, thereby helping to conceal the identity of the source. In the extreme case where $s=1$, we recover the standard ``push'' gossip protocol \cite{pittel1987spreading}, and show that it is inherently \emph{not} differentially private for large graphs. In contrast, we also show that, at the other end of the spectrum, choosing $s=0$ leads to \emph{optimal privacy guarantees} among all gossip protocols. 

More generally, we determine  \emph{matching upper and lower bounds} on the privacy guarantees of gossip protocols.
Essentially, our upper bounds on privacy are obtained by tightly lower bounding the probability that the source node contacts a curious node before another node does, and upper bounding the probability that this happens for a random node fixed in advance, in a way that holds for all gossip algorithms.
Remarkably, despite the fact that the source node always has a non-negligible probability of telling the rumor to a curious node first, our results highlight the fact that setting $s=0$ leads to strong privacy guarantees in several regimes, including the pure $(\epsilon,0)$-DP as well as prediction uncertainty.

It turns out that, although achieving optimal privacy guarantees, choosing $s=0$ leads to a very slow spreading time (\emph{log-linear} in the number of nodes $n$). This highlights an interesting tension between \emph{privacy} and \emph{spreading time}: the two extreme values for the muting parameter $s$ recover the two extreme points of this trade-off. We then show that more balanced trade-offs can be achieved: appropriate choices of the muting parameter lead to gossip protocols that 
are \emph{near-optimally private} with a spreading time that is \emph{logarithmic} in the size of the graph. In particular, the trade-off between privacy and speed shows up in the constants but, surprisingly, some choices of the parameter lead to protocols that achieve an optimal order of magnitude for both aspects. Our results on this trade-off are summarized in Table~\ref{tab:privacy_speed_tradeoffs}: for a proportion $f/n$ of curious nodes, one can see that setting the muting parameter $s= f/n$ achieves almost optimal privacy (up to a factor $2$) while being substantially faster than $s=0$ (optimal up to a factor $f/n$). Similarly, if one wants to achieve $(0,\delta_0)$-differential privacy with $\delta_0 > 2f/n$, then it is possible to set $s = \delta_0 / 2$ and obtain a protocol that respects the privacy constraint with spreading time $\mathcal{O}(\log(n) / \delta_0)$.
From a technical perspective, these privacy results are obtained by showing that only a small fraction of the possible observations by curious nodes have different probabilities when two different nodes start with the gossip. This requires to precisely evaluate the probability of well-chosen worst-case sequences, which is generally hard as randomness is involved both when nodes decide to stop spreading the rumor (with probability $1 - s$) and when they choose who to communicate with. Our parameterized gossip protocol can be seen as a population protocol~\cite{angluin2008fast}, and we prove its speed by analyzing its mean dynamics and leveraging concentration inequalities to bound the deviations from the mean dynamics.

\begin{table}[t]
    \centering
    \begin{tabular}{c|c|c|c}
         & Muting param. & $\delta$ ensuring $(0, \delta)$-DP & Spreading time\\
         \hline
         \begin{tabular}{cl} Standard push\\
         (minimal privacy, maximal speed)\end{tabular} & $s = 1$ &  1 & $\mathcal{O}(\log n)$\\
         \hline
         \begin{tabular}{c l} Muting after infecting\\
         (maximal privacy, minimal speed)\end{tabular} & $s=0$ & $\frac{f}{n}$ & $\mathcal{O}\left(n \log n\right)$\\
         \hline
         \begin{tabular}{c c} Generic parameterized gossip \\ (privacy vs. speed trade-off)\end{tabular} & $0 < s < 1$ & $s + (1 - s) \frac{f}{n}$ & $\mathcal{O}\left(\log (n) / s\right)$
    \end{tabular}
    \caption{Summary of results to illustrate the tension between privacy and speed. $n$ is the total number of nodes and $f/n$ is the fraction of curious nodes in the graph. $\delta\in[0,1]$ quantifies differential privacy guarantees (smaller is better). Spreading time is asymptotic in $n$.}
    \label{tab:privacy_speed_tradeoffs}
\end{table}

We support our theoretical findings by an empirical study of our parameterized gossip protocols. The results show that appropriate choices of $s$ lead to protocols that are very robust against classical source location attacks (such as maximum a posteriori estimates) while spreading the information almost as fast as the standard (and non-private) push protocol.
Crucially, we observe that our differential privacy guarantees are very well aligned with the ability to withstand attacks that leverage background information, e.g., targeting known activists or people who have been to certain places. 

The rest of the paper is organized as follows.
We first discuss related work and formally introduce our concept of differential privacy for gossip. Then, we study two extreme cases of our parameterized gossip protocol: the standard push protocol, which we show is not private, and a privacy-optimal but slow protocol. This leads us to investigate how to better control the trade-off between speed and privacy. Finally, we present our empirical study and conclude by discussing open questions.

For pedagogical reasons, we keep our model relatively simple to avoid unnecessary technicalities in the derivation and presentation of our results. For completeness, we discuss the impact of possible extensions (e.g., information observed by the adversary, malicious behavior, termination criterion) in Appendix~\ref{app:remarks}. For space limitations, some detailed proofs are also deferred to the appendix. 

\section{Background and Related Work}
\label{sec:related_work}

\subsection{Gossiping}
The idea of disseminating information in a distributed system by having each node \emph{push} messages to a randomly chosen neighbor, initially coined the  \emph{random phone-call model},  dates back to even before the democratization of the Internet \cite{pittel1987spreading}.
Such protocols, later called \emph{gossip}, \emph{epidemic} or \emph{rumor spreading},  were for instance applied  to ensure the consistency of a replicated database system \cite{demers1987epidemic}.
They have gained even more importance when argued to model spreading of infectious diseases \cite{SIS} and information dissemination in social networks \cite{doerr2011social, giakkoupis2015privacy}. Gossip protocols can also be used to compute aggregate queries on a database distributed across the nodes of a network \cite{Kempe2003a,boyd2006randomized}, 
 and have recently become popular in federated machine learning \cite{kairouz2019advances} to optimize cost functions over data distributed across a large set of peers
\cite{Duchi2012a,colin2016gossip,vanhaesebrouck2017decentralized,Koloskova2019}.
 Gossip protocols differ according to their interaction schemes, i.e., \emph{pull} or \emph{push}, sometimes combining both  \cite{karp2000randomized,kowalski2013estimating,acan2017push}. 

In this work, we focus on the classical \emph{push} form in the standard case of a \emph{complete} graph with $n$ nodes (labeled from $0$ to $n - 1$). We now define its key communication primitive. Denoting by $I$ the set of informed nodes, \verb+tell_gossip+$(i, I)$ allows an informed node $i\in I$ to tell the information to another node $j\in\{0, ..., n - 1\}$ chosen uniformly at random. \verb+tell_gossip+$(i, I)$ returns $j$ (the node that received the message) and the updated $I$ (the new set of informed nodes that includes $j$).
Equipped with this primitive, we can now formally define the class of gossip protocols that we consider in this paper.

\begin{definition}[Gossip protocols]
\label{def:push_algo}
A gossip protocol on a complete graph is one that (a) terminates almost surely, (b) ensures that at the end of the execution the set of informed nodes $I$ is equal to $\{0, ..., n-1\}$, and (c) can modify $I$ only through calls to \verb+tell_gossip+.
\end{definition}

\subsection{Locating the Source of the Gossip}
\label{sec:finding}
Determining the source of a gossip is an active research topic, especially given the potential applications to epidemics and social networks, see \cite{jiang2017identifying} for a recent survey. Existing approaches have focused so far on building \emph{source location attacks} that compute or approximate the maximum likelihood estimate of the source given some observed information. Each approach typically assumes a specific kind of graphs (e.g., trees, small world, \emph{etc}.), dissemination model and observed information. In \emph{rumor centrality} \cite{shah2011rumors},
the gossip communication graph is assumed to be fully observed and the goal is to determine the \emph{center} of this graph to deduce the node that started the gossip. 
Another line of work studies the setting in which some nodes are  \emph{curious sensors} that inform a central entity when they receive a message \cite{pinto2012locating}.
Gossiping is assumed to happen at random times and the source node is estimated by comparing the different timings at which information reaches the sensors.
The proposed attack is natural in trees but does not generalize to highly connected graphs.
The work of \cite{fanti2017hiding} focuses on hiding the source instead of locating it. The observed information is a snapshot of who has the rumor at a given time. A specific dissemination protocol is proposed to hide the source but the privacy guarantees only hold for tree graphs.

We emphasize that the privacy guarantees (i.e., the probability not to be detected) that can be derived from the above work only hold under the specific attacks considered therein. Furthermore, all approaches rely on maximum likelihood and hence assume a uniform prior on the probability of each node to be the source. The guarantees thus break if the adversary knows that some of the nodes could not have started the rumor, or if he is aware that the protocol is run twice from the same source.

We note that other problems at the intersection of gossip protocols and privacy have been investigated in previous work, such as preventing unintended recipients from learning the rumor \cite{georgiou2011confidential}, and hiding the initial position of agents in a distributed system \cite{gotfryd2017location}.


\subsection{Differential Privacy}
\label{sec:related_priv}
While we borrow ideas from the approaches mentioned above  (e.g., we assume that a subset of nodes are curious sensors as in \cite{pinto2012locating}), we aim at studying the fundamental limits of \emph{any} source location attack by measuring the amount of information leaked by a gossip scheme about the identity of the source. For this purpose, a general and robust notion of privacy is required.
\emph{Differential privacy} \cite{dwork2011differential,Dwork2014a} has emerged as a gold standard for it holds independently of any assumption on the model, the computational power,  or the background knowledge that the adversary may have. Differentially private protocols have been proposed for numerous problems in the fields of databases, data mining and machine learning: examples include computing aggregate and linear counting queries \cite{Dwork2014a}, releasing and estimating graph properties \cite{graph1, sun2019analyzing}, clustering \cite{clustering}, empirical risk minimization \cite{chaudhuri2011} and deep learning \cite{Abadi2016}.

In this work, we consider the classic version of differential privacy which involves two parameters $\epsilon,\delta\geq0$ that quantify the privacy guarantee \cite{dwork2006our}. More precisely, given any two databases $\mathcal{D}_1$ and $\mathcal{D}_2$ that differ in at most one record,
a randomized information release protocol $\mathcal{P}$ is said to guarantee $(\epsilon, \delta)$-differential privacy if for any possible output $S$:
\begin{equation}
\label{eq:classicdp}
p(\mathcal{P}(\mathcal{D}_1) \in S) \leq e^\epsilon p(\mathcal{P}(\mathcal{D}_2) \in S) + \delta,
\end{equation}
where $p(E)$ denotes the probability of event $E$. Parameter $\epsilon$ places a bound on the ratio of the probability of any output when changing one record of the database, while parameter $\delta$ is assumed to be small and allows the bound to be violated with small probability. When $\epsilon=0$, $\delta$ gives a bound on the total variation distance between the output distributions, while $\delta=0$ is sometimes called ``pure'' $\epsilon$-differential privacy.
DP guarantees hold regardless of the adversary and its background knowledge about the records in the database.
In our context, the background information could be the knowledge that the source is among a subset of all nodes. Robustness against such background knowledge is crucial in some applications, for instance when sharing secret information that few people could possibly know or when the source of an epidemic is known to be among people who visited a certain place. 
Another key feature of differential privacy is \emph{composability}: if $(\epsilon, \delta)$-differential privacy holds for a release protocol, then querying this protocol two times about the same dataset satisfies $(2 \epsilon, 2 \delta)$-differential privacy. This is important for rumor spreading as it enables to quantify privacy when the source propagates multiple rumors that the adversary can link to the same source (e.g., due to the content of the message). 
We will see in Section~\ref{sec:exp} that these properties are essential in practice to achieve robustness to concrete source location attacks.

Existing differentially private protocols typically introduce additional \emph{perturbation} (also called \emph{noise}) to hide critical information \cite{Dwork2014a}. In contrast, an original aspect of our work is that we will solely rely on the \emph{natural} randomness and limited observability brought by gossip protocols to guarantee differential privacy.

\section{A Model of Differential Privacy for Gossip Protocols}
\label{sec:model}

Our first contribution is a precise mathematical framework for studying the fundamental privacy guarantees of gossip protocols. We formally define the inputs of the gossip protocols introduced in Definition~\ref{def:push_algo}, 
the outputs observed by the adversary during their execution, and the privacy notions we investigate. 
To ease the exposition, we adopt the terminology of information dissemination, but we sometimes illustrate the ideas in the context of epidemics.

\subsection{Inputs and Outputs}
As described in Section~\ref{sec:related_priv}, differential privacy is a probabilistic notion that measures the privacy guarantees of a protocol based on the variations of its \emph{output} distribution for a change in its \emph{input}. In this paper, we adapt it to our gossip context. We first formalize the \emph{inputs} and \emph{outputs} of gossip protocols, in the case of a \emph{single piece of information to disseminate} (multiple pieces can be addressed through composition, see Section~\ref{sec:related_priv}). At the beginning of the protocol, a single node, the source, has the information (the gossip, or rumor). This node defines the input of the gossip protocol, and it is the actual ``database'' that we want to protect. Therefore, in our context, input databases in Equation~\eqref{eq:classicdp} have only $1$ record, which contains the identity of the source (an integer between $0$ and $n-1$). Therefore, all possible input databases differ in at most one record, and differential privacy aims at protecting the content of the database, i.e., which node started the rumor.

We now turn to the outputs of a gossip protocol.
The execution of a protocol generates an ordered sequence $S_{\rm omni}$ of pairs $(i, j)$ of calls to \verb+tell_gossip+ where $(S_{\rm omni})_t$ corresponds to the $t$-th time the \verb+tell_gossip+ primitive has been called, $i$ is the node on which \verb+tell_gossip+ was used and $j$ the node that was told the information. If several calls to \verb+tell_gossip+ happen simultaneously, ties are broken arbitrarily. We assume that the messages are received in the same order that they are sent. This protocol can thus be seen as an epidemic population protocol model~\cite{angluin2008fast} in which nodes interact using \verb+tell_gossip+.
The sequence $S_{\rm omni}$ corresponds to the output that would be observed by an omniscient entity who could eavesdrop on all communications. It is easy to see that, for any execution, the source can be identified exactly from $S_{\rm omni}$ simply by retrieving $(S_{\rm omni})_0$.

In this work, we focus on adversaries that monitor a set of \emph{curious nodes} $\mathcal{C}$ of size $f$, i.e. they observe all communications involving a curious node. This model, previously introduced in \cite{pinto2012locating}, is particularly meaningful in large distributed networks: while it is unlikely that an adversary can observe the full state of the network at any given time, compromising or impersonating a subset of the nodes appears more realistic. The number of curious nodes is directly linked with the \emph{release mechanism} of DP: while the final state of the system is always the same (everyone knows the rumor), the information released depends on which messages were received by the curious nodes during the execution.
Formally, the output disclosed to the adversary during the execution of the protocol, i.e., the information he can use to try to identify the source, is a subsequence of $S_{\rm omni}$ as defined below.
\begin{assumption}
The sequence $S$ observed by the adversary through the (random) execution of the protocol is a (random) subsequence $S = \left( (S_{\rm omni})_t | (S_{\rm omni})_t = (i, j) \hbox{ with } j \in \mathcal{C} \right)$, that contains all messages sent to curious nodes. The adversary has access to the relative order of tuples in $S$, which is the same as in $S_{\rm omni}$, but not to the index $t$ in $S_{\rm omni}$.
\end{assumption}
It is important to note that the adversary does not know which messages were exchanged between non-curious nodes. In particular, he does not know how many messages were sent in total at a given time.
As we focus on complete graphs, knowing which curious node received the rumor gives no information on the source node. For a given output sequence $S$, we write $S_t = i$ to denote that the $t$-th \verb+tell_gossip+ call in $S$ originates from node $i$. Omitting the dependence on $S$, we also denote $t_i(j)$ the time at which node $j$ first receives the message (even for the source) and $t_d(j)$ the time at which $j$ first communicates with a curious node. 

The ratio $f / n$ of curious nodes determines the probability of the adversary to gather information (the more curious nodes, the more information leaks). For a fixed $f$, the adversary only becomes weaker as the network grows bigger. Since we would like to study adversaries with fixed power, unless otherwise noted we make the following assumption.

\begin{assumption}
The ratio of curious nodes $f/n$ is constant.
\end{assumption}
Finally, we emphasize that we do not restrict the computational power of the adversary.

\subsection{Privacy Definitions}
\label{subsec:private_gossip}
We now formally introduce our privacy definitions. The first one is a direct application of differential privacy (Equation~\ref{eq:classicdp}) for the inputs and outputs specified in the previous section. To ease notations, we denote by $I_0$ the source of the gossip (the set of informed nodes at time $0$), and  for any given $i \in \{0, ..., n - 1\}$, we denote by  $p_i(E) = p(E | I_0 = \{i\})$ the probability of event $E$ if node $i$ is the source of the gossip. The protocol is therefore abstracted in this notation. Denoting by $\mathcal{S}$ the set of all possible outputs, we say that a gossip protocol is $(\epsilon,\delta)$-differentially private if:
\begin{equation}
\label{eq:dp}
p_i(S) \leq e^\epsilon p_j(S) + \delta, \ \ \forall S \subset \mathcal{S}, \ \forall i,j \in \{0, ..., n- 1\},
\end{equation}
where $p(S)$ is the probability that the output belongs to the set $S$. This formalizes a notion of \emph{source indistinguishability} in the sense that any set of output which is likely enough to happen if node $i$ starts the gossip (say, $p_i(S) \geq \delta$) is almost as likely (up to a $e^\epsilon$ multiplicative factor) to be observed by the adversary regardless of the source. Note however that when $\delta>0$, this definition can be satisfied for protocols that release the identity of the source (this can happen with probability $\delta$). To capture the behavior under unlikely events, we also consider the complementary notion of $c$-\emph{prediction uncertainty} for $c > 0$, which is satisfied  if for a uniform prior $p(I_0)$ on source nodes and any $i \in \{0, ..., n- 1\}$:
\begin{equation}
\label{eq:source_un}
p(I_0 \neq \{i \} | S) / p(I_0 = \{ i \} | S) \geq c,
\end{equation}
for any  $S \subset \mathcal{S}$ such that $p_i(S) > 0$. Prediction uncertainty guarantees that no observable output $S$ (however unlikely) can identify a node as the source with large enough probability: it ensures that the probability of success of any source location attack is upper bounded by $1/(1+c)$. This holds in particular for the maximum likelihood estimate. Prediction uncertainty does not have the robustness of differential privacy against background knowledge, as it assumes a uniform prior on the source. While it can be shown that $(\epsilon, 0)$-DP with $\epsilon>0$ implies prediction uncertainty, the converse is not true. Indeed, prediction uncertainty is satisfied as soon as no output identifies any node with enough probability, without necessarily making all pairs of nodes indistinguishable as required by DP. We will see that prediction uncertainty allows to rule out some naive protocols that have nonzero probability of generating sequences which reveal the source with certainty. 

Thanks to the symmetry of our problem, we consider without loss of generality that node $0$ starts the rumor ($I_0 = \{0\}$) and therefore we will only need to verify Equations~\eqref{eq:dp} and~\eqref{eq:source_un} for $i=0$ and $j=1$.

\section{Extreme Privacy Cases}
\label{sec:extreme_privacy}

In this section, we study the fundamental limits of gossip in terms of privacy. To do so, we parameterize 
gossip protocols by a muting parameter $s\in[0,1]$, as depicted in Algorithm~\ref{algo:gen_sync_push}. We thereby capture, within a generic framework,  a large family of protocols that fit Definition~\ref{def:push_algo} and work as follows. They maintain a set $A$ of \emph{active nodes} (initialized to the source node) which spread the rumor asynchronously and in parallel: this is modeled by the fact that at each step of the protocol, a randomly selected node $i\in A$ invokes the \verb+tell_gossip+ primitive to send the rumor to another node (which in turn becomes active), while $i$ also stays active with probability $s$. This is illustrated in Figure~\ref{fig:illustr}.
The muting parameter $s$ can be viewed as a randomized version of \emph{fanout} in \cite{eugster2004epidemic}.\footnote{Unlike in the classic fanout, nodes start to gossip again each time they receive a message instead of deactivating permanently.}
Algorithm~\ref{algo:gen_sync_push} follows the population protocol model~\cite{angluin2008fast}, and is also related to the SIS epidemic model \cite{SIS} but in which the rumor never dies regardless of the value of $s\in[0,1]$ (there always remain some active nodes). Although we present it from a centralized perspective, we emphasize that Algorithm~\ref{algo:gen_sync_push} is asynchronous and can be implemented by having active nodes wake up following a Poisson process.

In the rest of this section, we show that extreme privacy guarantees are obtained for extreme values of the muting parameter $s$.

\begin{figure*}[t]
 \centering
    \begin{minipage}{.5\textwidth}
        \centering
\begin{algorithm}[H]
	\caption{Parameterized Gossip}
	\label{algo:gen_sync_push}
	\begin{algorithmic}[1]
		\REQUIRE
		$n$ \COMMENT{Number of nodes}, $k$ \COMMENT{Source node}, $s$ \COMMENT{Probability for a node to remain active}
        \ENSURE $I=\{0,\dots,n-1\}$ \COMMENT{All nodes are informed}
		\STATE $I \leftarrow \{ k \}$, $A \leftarrow \{ k \}$
		\WHILE{$| I | < n$}
		\STATE{Sample $i$ uniformly at random from $A$}
		\STATE $A \leftarrow A \setminus \{i\}$ with probability $1-s$
		\STATE{$j, I \leftarrow$ \verb+tell_gossip+$(i, I)$, $A\leftarrow A\cup\{j\}$}
		\ENDWHILE
	\end{algorithmic}
\end{algorithm}
    \end{minipage}%
    \begin{minipage}{0.5\textwidth}
        \centering
\includegraphics[width=0.72\linewidth]{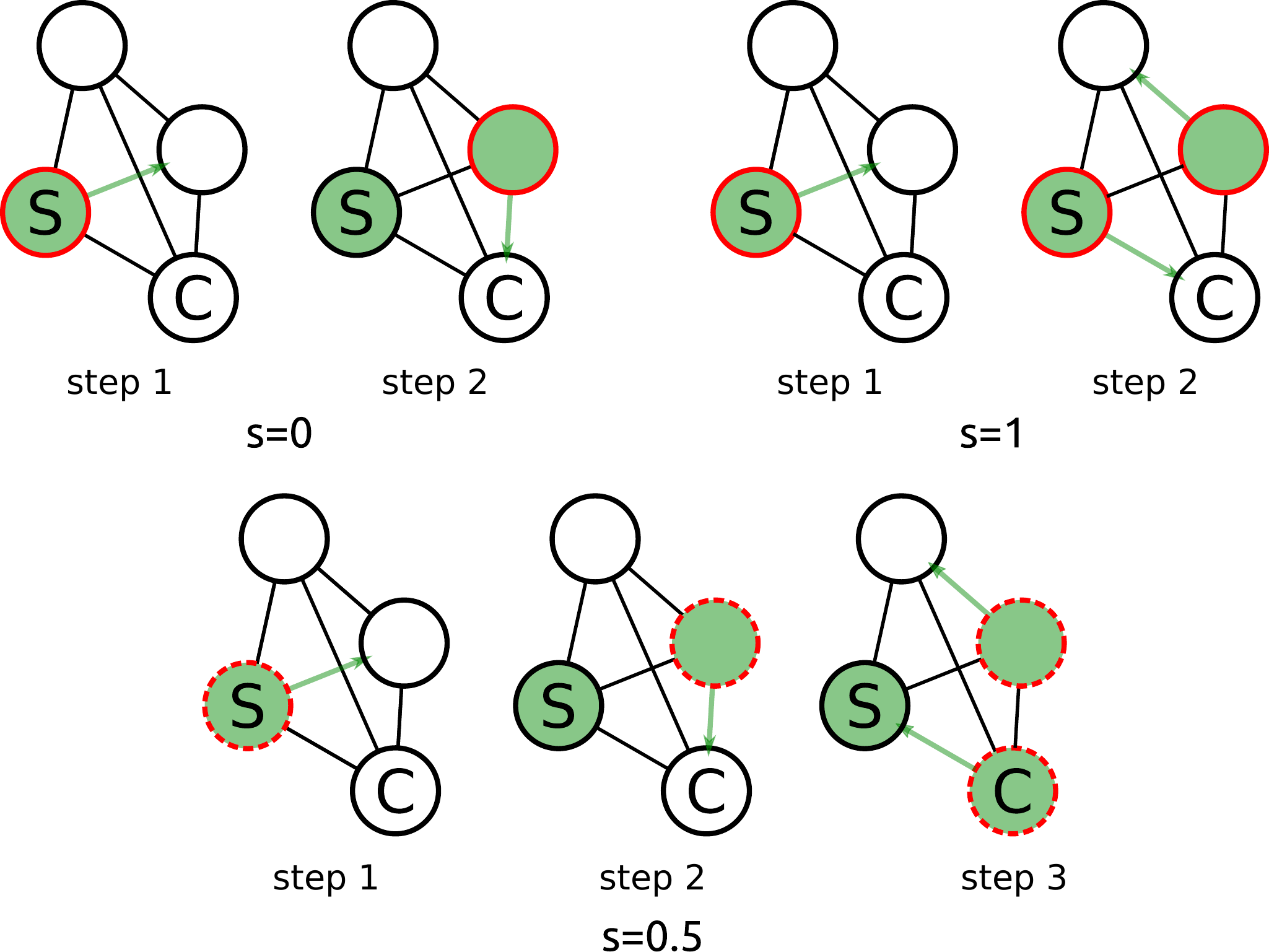}
    \end{minipage}
\caption{\emph{Left:} Parameterized Gossip. \emph{Right:} Illustration of the role of muting parameter $s$. \texttt{S} indicates the source and \texttt{C} a curious node. Green nodes know the rumor, and red circled nodes are active. When $s=0$, there is only one active node at a time, which always stops emitting after telling the gossip. 
In the case $s=1$, nodes always remain active once they know the rumor
(this is the standard push gossip protocol \cite{pittel1987spreading}). When $0<s<1$, each node remains active with probability $s$ after each communication.
}
\label{fig:illustr}
\end{figure*}

\subsection{Standard Push has Minimal Privacy}
\label{sec:push_not_private}
The natural case to study first in our framework is when the muting parameter is set to $s=1$: this corresponds to the standard push protocol \cite{pittel1987spreading} in which nodes always keep emitting after they receive the rumor. Theorem~\ref{cor:impossibility_sync} shows that, surprisingly, the privacy guarantees of this protocol become arbitrarily bad as the size of the graph increases (keeping the fraction of curious nodes constant). 


\begin{theorem}[Standard push is not differentially private]
\label{cor:impossibility_sync}
If Algorithm~\ref{algo:gen_sync_push} with $s=1$ guarantees $(\epsilon, \delta)$-DP for all values of $n$ and constant $\epsilon<\infty$, then $\delta = 1$.
\end{theorem}

This result may seem counter-intuitive at first since one could expect that it would be more and more difficult to locate the source when the size of the graph increases. Yet, since the ratio of curious nodes is kept constant, this result comes from  the fact that the event $\{t_d(0) \leq t_i(1)\}$ (node $0$ communicates with a curious node before node $1$ gets the message) becomes more and more likely as $n$ grows, hence preventing any meaningful differential privacy guarantee when $n$ is large enough. The proof is in Appendix~\ref{app:proofs_faster_push}.
Theorem~\ref{cor:impossibility_sync} clearly highlights the fact that the standard gossip protocol ($s=1$) is not differentially private in general. We now turn to the other extreme case, where the muting parameter $s=0$.

\subsection{Muting After Infecting has Maximal Privacy}
\label{sec:optimal-alg}

We now study the privacy guarantees of generic Algorithm~\ref{algo:gen_sync_push} when $s=0$. In this protocol, nodes forward the rumor to exactly one random neighbor when they receive it and then stop emitting until they receive the rumor again. Intuitively, this is good for privacy: the source changes and it is quickly impossible to recover which node started the gossip (as initial conditions are quickly forgotten). In fact, once the source tells the rumor once, the state of the system (the set of active nodes, which in this case is only one node) is completely independent from the source. A similar idea was used in the protocol introduced in \cite{fanti2017hiding}.

The following result precisely quantifies the privacy guarantees of Algorithm~\ref{algo:gen_sync_push} with parameter $s=0$ and shows that it is \emph{optimally private} among all gossip protocols (in the precise sense of Definition~\ref{def:push_algo}).


\begin{theorem}
\label{thm:push_lb}
Let $\epsilon\geq 0$. For muting parameter $s=0$, Algorithm~\ref{algo:gen_sync_push} satisfies $(\epsilon, \delta)$-differential privacy with $\delta = \frac{f}{n} \left( 1 - \frac{e^\epsilon - 1}{f} \right)$ and $c$-prediction uncertainty with  $c = \frac{n}{f +1} - 1$.
Furthermore, these privacy guarantees are optimal among all gossip protocols.
\end{theorem}


\begin{proof}[Proof of Theorem~\ref{thm:push_lb}] We start by proving the fundamental limits on the privacy of any gossip protocol, and then prove matching guarantees for Algorithm~\ref{algo:gen_sync_push} with $s=0$.

\textbf{(Fundamental limits in privacy)} Proving a lower bound on the differential privacy parameters can be achieved by finding a set of possible outputs $S$ (here, a set of ordered sequences) such that $p_0(S) \geq p_1(S)$. Indeed, a direct application of the definition of Equation~\eqref{eq:dp} yields that given any gossip protocol, $S \subset \mathcal{S}$ and $w_0, w_1 \in \mathbb{R}$ such that $w_0 \leq p_0(S)$ and $p_1(S) \leq w_1$, if the protocol satisfies $(\epsilon, \delta)$ differential privacy then $\delta \geq w_0 - e^\epsilon w_1$. The proofs need to consider all the messages sent and then distinguish between the ones that are disclosed (sent to curious nodes) and the ones that are not.  

Since $I = \{ 0 \}$ then \verb+tell_gossip+ is called for the first time by node $0$ and it is called at least once otherwise the protocol terminates with $I = \{ 0 \}$, violating the conditions of Definition~\ref{def:push_algo}. We denote by $S^{(0)}$ the set of output sequences such that $S_0 = 0$ (i.e., $0$ is the first to communicate with a curious node). We also define the event $T_0^c = \{ t_d(0) \neq 0\}$  (the source does not send its first message to a curious node). For all $i \notin \mathcal{C} \cup \{0\}$, we have that $p_0(S_0 = i | T_0^c) \leq p_0(S_0 = 0 | T_0^c)$ since $p_0(A_1 = \{0\}) = p_0(i \in A_1)$, where $A_1$ is the set of active nodes at time $1$. From this inequality we get
\begin{align*}
\textstyle\sum_{i \notin \mathcal{C}} p_0(S_0 = 0 | T_0^c) \geq \sum_{i \notin \mathcal{C}} p_0(S_0 = i | T_0^c) =1 \geq \sum_{i \notin \mathcal{C}} p_0(S_0 = 1 | T_0^c),
\end{align*}
where the equality comes from the fact that $S_0 = i$ for some $i \notin \mathcal{C}$. The second inequality comes from the fact that $p_j(S_0 = {i} | T_0^c) = p_j(S_0 = {k} | T_0^c)$ for all $i,k \neq j$. Therefore, we have $p_0(S_0 = 0 | T_0^c) \geq \frac{1}{n - f}$ and $p_0(S_0 = 1 | T_0^c) \leq \frac{1}{n - f}$.
Combining the above expressions, we  derive the probability of $S^{(0)}$ when $0$ started the gossip. We write $p_0(S^{(0)}) = p_0(S^{(0)}, t_d(0) = 0) + p_0(S^{(0)}, T_0^c )$ and then, since $p_0\big(S^{(0)} | t_d(0) = 0\big) = 1$:
\begin{align*}
p_0\big(S^{(0)}\big) = p_0\big(t_d(0) = 0\big) p_0\big(S^{(0)} | t_d(0) = 0\big) + p_0\big(S^{(0)} | T_0^c \big)p_0\big(T_0^c \big) \geq \frac{f}{n} + \frac{1}{n - f}\Big(1 - \frac{f}{n}\Big)
\end{align*}
In the end, $p_0(S^{(0)}) \geq \frac{f}{n} + \frac{1}{n}$. If node $1$ initially has the message, we do the same split and obtain $p_1(S^{(0)} | t_d(0) = 0) = 0$ and so $p_1(S^{(0)}) = p_1(T_0^c)p_1(S^{(0)} | T_0^c) \leq \frac{1}{n}$ .

The upper bound on prediction uncertainty is derived using the same quantities:
\begin{align*}
\frac{p(I_0 \neq 0 | S^{(0)})}{p(I_0 = 0 | S^{(0)})} = \sum_{i \notin C \cup \{0\}} \frac{p_i(S^{(0)})}{p_0(S^{(0)})} \leq (n - f - 1) \frac{p_1(S^{(0)})}{p_0(S^{(0)})}
\leq \frac{n - f - 1}{f + 1} = \frac{n}{f + 1} - 1.
\end{align*}
Note that we have never assumed that curious nodes knew how many messages were sent at a given point in time. We have only bounded the probability that the source is the first node that sends a message to curious nodes.

\textbf{(Matching guarantees for Algorithm~\ref{algo:gen_sync_push} with $s=0$)}
For this protocol, the only outputs $S$ such that $p_0(S) \neq p_1(S)$ are those in which $t_d(0) = 0$ or $t_d(1) = 0$. We write:
\begin{equation*}
p_0(S_0 = 0) = p_0(t_d(0) = 0)p_0(S_0 = 0 | t_d(0) = 0) +  p_0(T_0^c)p_0(S_0 = 0 | T_0^c).
\end{equation*}
For any $i \notin \mathcal{C}$ where $\mathcal{C}$ is the set of curious nodes, we have that $p_0(S_0 = 0 | T_0^c) = p_0(S_0 = i | T_0^c) = \frac{1}{n -f}$. Indeed, given that $t_d(0) \neq 0$, the node that receives the first message is selected uniformly at random among non-curious nodes, and has the same probability to disclose the gossip at future rounds. Plugging into the previous equation, we obtain:
\begin{equation*}
p_0(S_0 = 0) = \frac{f}{n} +  \Big(1 - \frac{f}{n}\Big) \frac{1}{n - f} = \frac{f + 1}{n}.
\end{equation*}
For any other node $i\notin \mathcal{C} \cup \{0\}$,
$p_0(S_0 = i) = p_0(T_0^c) p_0(S_0 = i | T_0^c) = \frac{1}{n}$ because $p_0(S_0 = i | t_d(0) = 0) = 0$. Combining these results we get $p_0(S^{(0)}) \leq e^\epsilon p_1(S^{(0)}) + \delta$
for any $\epsilon>0$ and $\delta=\frac{f}{n} ( 1 - \frac{e^\epsilon - 1}{f})$. By symmetry, we make a similar derivation for $S^{(1)}$.

To prove the prediction uncertainty result, we use the differential privacy result with $e^\epsilon = f + 1$ (and thus $\delta = 0$) and write that for any $S \in \mathcal{S}$:
\begin{equation*}
\frac{p(I_0 \neq 0 | S)}{p(I_0 = 0 | S)} = \sum_{i \notin C \cup \{0\}} \frac{p_i(S)}{p_0(S)} \geq (n - f - 1)e^{-\epsilon} = \frac{n}{f + 1} - 1.\qedhere
\end{equation*}
\end{proof}

Theorem~\ref{thm:push_lb} establishes \emph{matching upper and lower bounds} on the privacy guarantees of gossip protocols. More specifically, it shows that setting the muting parameter to $s=0$ provides strong privacy guarantees that are in fact optimal.
Note that in the regime where $\epsilon=0$ (where DP corresponds to the total variation distance), $\delta$ cannot be smaller than the proportion of curious nodes. This is rather intuitive since the source node has probability at least $f/n$ to send its first message to a curious node.
However, one can also achieve differential privacy with $\delta$ much smaller than $f/n$ by trading-off with $\epsilon>0$. In particular, the \emph{pure} version of differential privacy ($\delta = 0$) is attained for $\epsilon \approx \log f$, which provides good privacy guarantees when the number of curious nodes is not too large. Furthermore,
even though the probability of disclosing \emph{some} information is of order $f / n$, prediction uncertainty guarantee shows that an adversary with uniform prior always has a high probability of making a mistake when predicting the source. Crucially, these privacy guarantees are made possible by the \emph{natural} randomness and partial observability of gossip protocols. 



\begin{remark}[Special behavior of the source]
\label{rem:naive}
A subtle but key property of Algorithm~\ref{algo:gen_sync_push} is that the source follows the same behavior as other nodes. To illustrate how violating this property may give away the source, consider this natural protocol: the source node transmits the rumor to one random node and stops emitting, then standard push (Algorithm~\ref{algo:gen_sync_push} with $s=1$) starts from the node that received the information. While this \emph{delayed start gossip protocol} achieves optimal differential privacy in some regimes, it is fundamentally flawed. In particular, it does not guarantee prediction uncertainty in the sense that $c \rightarrow 0$ as the graph grows. Indeed, the adversary can identify the source with high probability by detecting that it communicated only once and then stopped emitting for many rounds. We refer to 
 Appendix~\ref{app:weak_privacy_protocol} 
  for the formal proof.
\end{remark}

\section{Privacy vs. Speed Trade-offs}
\label{sec:faster_push}

While choosing $s=0$ achieves optimal privacy guarantees, an obvious drawback is that it leads to a very slow protocol since only one node can transmit the rumor at any given time. 
It is easy to see that the number of gossip operations needed 
to inform all nodes can be reduced to the time needed for the classical coupon collection problem: it takes $\mathcal{O}(n \log n)$ communications to inform all nodes with probability at least $1 - 1/n$  \cite{erdHos1961classical}. As this protocol performs exactly one communication at any given time, it needs time $\mathcal{O}(n \log n )$ to inform all nodes with high probability. This is in stark contrast to the standard push gossip protocol ($s=1$) studied in Section~\ref{sec:push_not_private} where all informed nodes can transmit the rumor in parallel, requiring only time $\mathcal{O}(\log n)$~\cite{frieze1985shortest}.

These observations motivate the exploration of the privacy-speed trade-off (with parameter $0<s<1$).
We first show below that nearly optimal privacy can be achieved for small values of $s$. Then, we study the spreading time and show that the $\mathcal{O}(\log n)$ time of the standard gossip protocol also holds for $s>0$, leading to a sweet spot in the privacy-speed trade-off.

\subsection{Privacy Guarantees}
\label{sec:privacy_fast_private_gossip}

Theorem~\ref{thm:sync_push_s} conveys a $(0, \delta)$-differential privacy result, which means that apart from some unlikely outputs that may disclose the identity of the source node, most of these outputs actually have the same probability regardless of which node triggered the dissemination.
We emphasize that the guarantee we obtain holds for any graph size with fixed proportion $f/n$ of curious nodes.

\begin{theorem}[Privacy guarantees for $s<1$]
\label{thm:sync_push_s}
For $0 < s < 1$ and any fixed $r \in \mathbb{N}^*$, Algorithm~\ref{algo:gen_sync_push} with muting parameter $s$ guarantees $(0, \delta)$-differential privacy with: 
\begin{equation*}
\delta = 1 - (1 - s)\sum_{k=0}^\infty s^k \left(1 - \frac{f}{n}\right)^{k+1} \leq 1 - (1 - s^r)\left(1 - \frac{f}{n}\right)^{r}.
\end{equation*}
For example, choosing $r=1$ leads to $\delta \leq s + (1 - s)\frac{f}{n}$, as reported in Table~\ref{tab:privacy_speed_tradeoffs}. Slightly tighter bounds can be obtained, but this is enough already to recover optimal guarantees as $s \rightarrow 0$.
\end{theorem}

\begin{proof}
We first consider that $S$ is such that $t_d(0) \geq t_d(1)$. Then, $p_0(S) \leq p_1(S)$ since node $0$ needs to receive the rumor before being able to communicate it to curious nodes, and Equation~\eqref{eq:dp} is verified. Suppose now that $S$ is such that $t_d(0) \leq t_d(1)$. In this case, we note $t_m$ the first time at which the source stops to emit (which happens with probability $1-s$ each time it sends a message). Then, we denote $F = \{t_d(0) \leq t_m\}$ (and $F^c$ its complement). In this case, $p_0(S | F^c) \leq p_1(S | F^c)$. Indeed, conditioned on $F^c$, $t_d(0) \geq t_i(0)$ if node $0$ is not the source and $t_d(0) \geq \max(t_m, t_i(0))$ if it is. Then, we can write:
\begin{align*}
p_0(S) &= p_0(S, F^c) + p_0(S, F) \leq p_1(S, F^c) + p_0(F) \leq p_1(S) + p_0(F).
\end{align*}
Denoting $T_f$ the number of messages after which the source stops emitting, we write: 
\begin{align*}
p_0(F) &= \sum_{k=1}^\infty p_0(T_f = k)p_0( F | T_f = k) = \sum_{k=0}^\infty (1 - s) s^k \Big(1 - \big(1 - \frac{f}{n}\big)^{k+1}\Big), \hbox{ for }  s>0.
\end{align*}
Note that we can also write for $k \geq 1$ that $p_0(F) = p_0(F, T_f \leq k) + p_0(F, T_f > k)$, and so:
\begin{equation*}
p_0(F) \leq (1 - s^k)\Big(1 - \big(1 - \frac{f}{n}\big)^{k}\Big) + s^k = 1 - (1 - s^k)\Big(1 - \frac{f}{n}\Big)^{k}. \qedhere
\end{equation*}
\end{proof}

The differential privacy guarantees given by Theorem~\ref{thm:sync_push_s} and the optimal guarantees of Theorem~\ref{thm:push_lb} are of the same order of magnitude when $s$ is of order $f/n$. Indeed, consider $\epsilon = 0$. Then, setting $r=1$ in Theorem~\ref{thm:sync_push_s} leads to an additive gap of $s(1 - f/n)$ between the privacy of Algorithm~\ref{algo:gen_sync_push} and the optimal guarantee, showing that one can be as close as desired to the optimal privacy as long as $s$ is chosen close enough to $0$. In particular, the ratio between the privacy of Algorithm~\ref{algo:gen_sync_push} and the lower bound is less than $2$ for all $s \leq f/n$. This indicates that the privacy guarantees are very tight in this regime. We also recover exactly the optimal guarantee of Theorem~\ref{thm:push_lb} in the case $s=0$ (without the ability to control the trade-off between $\epsilon$ and $\delta$). 
Importantly, we also show that Algorithm~\ref{algo:gen_sync_push} with $s<1$ satisfies prediction uncertainty, unlike the case where $s=1$.

\begin{theorem}
\label{thm:source_uncertainty_private_push}
Algorithm~\ref{algo:gen_sync_push} guarantees prediction uncertainty with $c=(1-\frac{f+1}{n})(1-s)$.
\end{theorem}

This result is another evidence that picking $s < 1$ allows to derive meaningful privacy guarantees. The proof can be found in  
 Appendix~\ref{app:proofs_faster_push}. 

\subsection{Spreading time}
\label{sec:diffusion_s}

We have shown that parameter $s$ has a significant impact on privacy, from optimal ($s=0$) to very weak ($s=1$) guarantees. Yet, $s$ also impacts the spreading time: the larger $s$, the more active nodes at each round. This is highlighted by the two extreme cases, for which the spreading time is already known and exhibits a large gap: $\mathcal{O}(\log n)$ for $s=1$ and $\mathcal{O}(n\log n)$ for $s=0$. To establish whether we can obtain a protocol that is both private and fast, we need to characterize the spreading time for the cases where $0<s<1$.

The key result of this section is to prove that the logarithmic speed of the standard push gossip protocol holds more generally for all $s>0$.  
This result is derived from the fact that the ability to forget does not prevent an \emph{exponential growth} phase. What changes is that the population of active nodes takes approximately $1/s$ rounds to double instead of $1$ for standard gossip.
For ease of presentation, we state below the result for the synchronous version of Algorithm~\ref{algo:gen_sync_push}, in which the notion of \emph{round} corresponds to iterating over the full set $A$.
A similar result (with an appropriate notion of rounds) can be obtained for the asynchronous version given in Algorithm~\ref{algo:gen_sync_push}.

\begin{theorem}
\label{thm:speed_s}
For a given $s > 0$ and for all $1 > \delta > 0$ and $C \geq 1$, there exists $n$ large enough such that the synchronous version of Algorithm~\ref{algo:gen_sync_push} with parameter $s$ sends at least $C n \log n$ messages in $6 C \log(n) / s$ rounds with probability at least $1 - \delta$.
\end{theorem}

\begin{proof}[Proof sketch]
The key argument of the proof is that the gossip process very closely follows its mean dynamics. After a transition phase of a logarithmic number of rounds, a constant fraction of the nodes (depending on $s$) remains active despite the probability to stop emitting after each communication. This ``determinism of gossip process" has been introduced in \cite{sanghavi2007gossiping}, but their analysis only deals with the case $s=1$. Our proof takes into account the nontrivial impact of nodes deactivation in the exponential and linear growth phase. Besides, we need to introduce and analyze a last phase, showing that with high probability the population never drops below a critical threshold of active nodes. The full proof is in 
 Appendix~\ref{app:speed}.
\end{proof}

Theorem~\ref{thm:speed_s} shows that generic gossip with $s>0$ still achieves a logarithmic spreading time even though nodes can stop transmitting the message. The $1/s$ dependence is intuitive since $1/s$ rounds are needed in expectation to double the population of active nodes (without taking collisions into account). Therefore, the exponential growth phase which usually takes time $\mathcal{O}(\log n)$ now takes time $\mathcal{O}(\log(n) / s)$ for $s < 1$.
To summarize, we have shown that one can achieve both fast spreading and near-optimal privacy, leading to the values presented in Table~\ref{tab:privacy_speed_tradeoffs} of the introduction.

\section{Empirical Evaluation}
\label{sec:exp}

In this section, we evaluate the practical impact of $s$ on the spreading time as well as on the robustness to source location attacks run by adversaries with background knowledge.

\subsection{Spreading Time}

\begin{figure*}[t]
\centering
\begin{subfigure}{.5\textwidth}
  \centering
  \includegraphics[width=.9\linewidth]{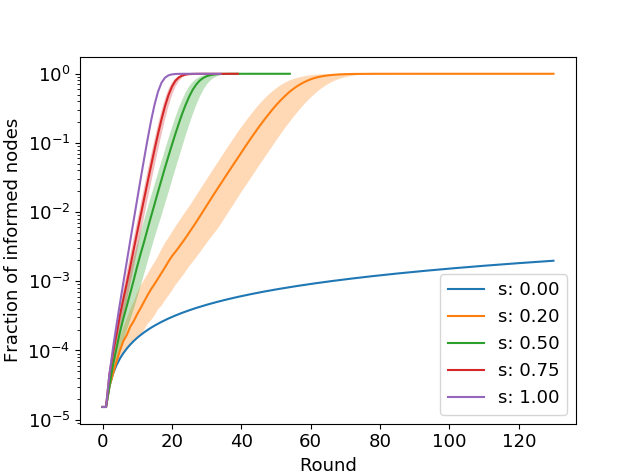}
\caption{Fraction of informed nodes}
\label{fig:fraction_informed}
\end{subfigure}%
\begin{subfigure}{.5\textwidth}
  \centering
  \includegraphics[width=.9\linewidth]{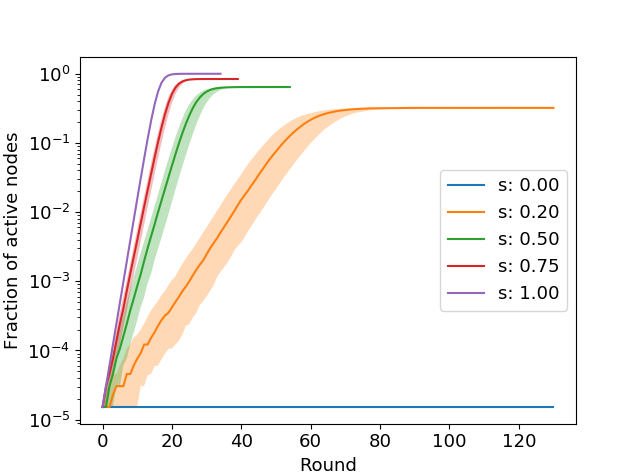}
\caption{Fraction of active nodes}
\label{fig:fraction_active}
\end{subfigure}%
\caption{Effect of parameter $s$ of Algorithm~\ref{algo:gen_sync_push} on the spreading time for a network of $n=2^{16}$ nodes. The curves represent median values and the shaded area represents the 10 and 90 percent confidence intervals over 100 runs. Each curve stops when all nodes are informed (and so the protocol terminates), except for $s=0$ since the protocol is very slow in this case.}
\label{fig:bounds}
\end{figure*}

To complement Theorem~\ref{thm:speed_s}, which proves logarithmic spreading time (asymptotic in $n$), we run simulations on a network of size $n=2^{16}$. The logarithmic spreading time for $s>0$ is clearly visible in Figure~\ref{fig:fraction_informed}, where we see that the gossip spreads almost as fast for $s=0.5$ that it does for $s=1$. We also observe that even when $s$ is small, the gossip remains much faster than for $s=0$. The results in Figure~\ref{fig:fraction_active} illustrate that the fraction of active nodes grows exponentially fast for all values of $s>0$ and then reaches a plateau when the probability of creating a new active node is compensated by the probability of informing an already active node. Empirically, this happens when the fraction of active nodes is of order $s$.

We note incidentally that gossip protocols are often praised for their robustness to lost messages~\cite{alistarh2010efficient,georgiou2013asynchronous}. While the protocol with $s=0$ does not tolerate a single lost message, setting $s>0$ improve the resilience thanks to the linear proportion of active nodes. The latter property makes it unlikely that the protocol stops because of lost messages as long as $s$ is larger than the probability of losing messages. Of course, the protocol remains somewhat sensitive to messages lost during the first few steps.

\subsection{Robustness Against Source Location Attacks}

\begin{figure*}[t]
\centering
\begin{subfigure}{.5\textwidth}
  \centering
  \includegraphics[width=.9\linewidth]{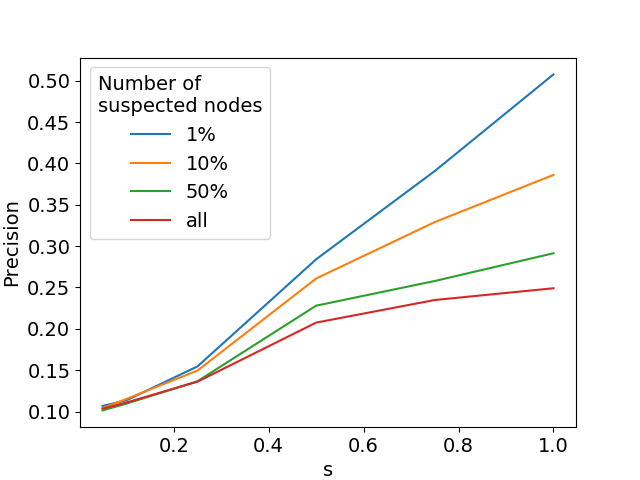}
\caption{ \centering Attack precision under prior information on the source}
\label{fig:prior_attack}
\end{subfigure}%
\begin{subfigure}{.5\textwidth}
  \centering
  \includegraphics[width=.9\linewidth]{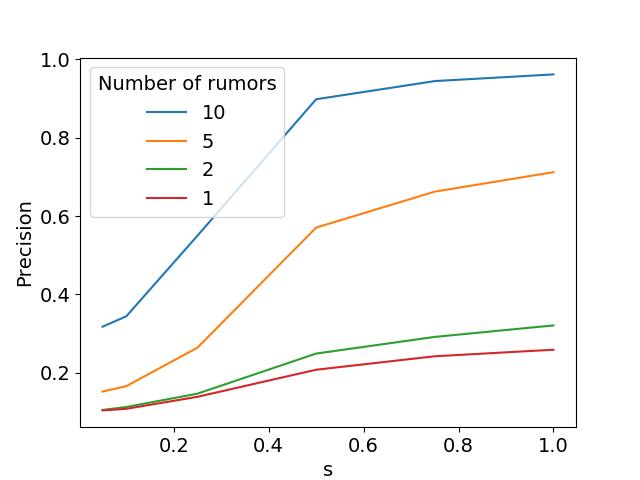}
\caption{ \centering Attack precision when the source spreads multiple rumors}
\label{fig:composition_attack}
\end{subfigure}%
\caption{Effect of parameter $s$ of Algorithm~\ref{algo:gen_sync_push} on the precision of source location attacks for a network of $n=2^{16}$ node with $10\%$ of curious nodes. Precision is estimated over 15,000 random runs.}
\label{fig:attacks}
\end{figure*}

Getting an intuitive understanding of the privacy guarantees provided by Theorem~\ref{thm:sync_push_s} is not straightforward, as often the case with differential privacy. Therefore, we illustrate the effect of the muting parameter on the guarantees of our gossip protocol by simulating concrete source location attacks. We consider two challenging scenarios where the adversary has some background knowledge: either 1) prior knowledge that the source belongs to a subset of the nodes, or 2) side information indicating that the same source disseminates multiple rumors.\vspace{5pt}

\noindent\textbf{Prior knowledge on the source.}
We first consider the case where the adversary is able to narrow down the set of suspected nodes. In this case we can design a provably optimal attack, as shown by the following theorem 
 (see Appendix~\ref{app:proof_mle} for the proof).

\begin{theorem}
\label{thm:mle}
If the adversary has a uniform prior over a subset $P$ of nodes, i.e., $p(I_0 = i) = p(I_0 = j)$ for all $i,j \in P$ and $p(I_0 = i) = 0$ for $i \notin P$, and for some output sequence $S$, $t_c$ is such that $S_{t_c} \in P$ and $S_t \notin P$ if $t < t_c$ , then $p(I_0 = S_{t_c} | S) \geq p(I_0 = i | S)$ for all $i$. 
\end{theorem}
Theorem~\ref{thm:mle} means that under a uniform prior over nodes in $P$, the attack in which curious nodes predict the source to be the first node in $P$ that communicates with them corresponds to the Maximum A Posteriori (MAP) estimator. The set $P$ represents the prior knowledge of the adversary: he knows for sure that the source belongs to $P$.

Figure~\ref{fig:prior_attack} shows the precision of this attack as a function of $s$ for varying degrees of prior knowledge. We see that, when $s$ is small, the prior knowledge does not improve the attack precision significantly, and that the precision remains very close to the probability that the source sends its first message to a curious node. This robustness to prior knowledge is consistent with the properties of differential privacy (see Section~\ref{sec:related_priv}). On the contrary, when $s$ is high (i.e., differential privacy guarantees are weak), the impact of the prior knowledge on the precision of the attack is much stronger.\vspace{5pt}

\noindent\textbf{Multiple dissemination.}
We investigate another scenario in which differential privacy guarantees can also provide robustness, namely when the adversary knows that the same source node disseminates multiple rumors.
In this setting, analytically deriving an optimal attack is very difficult. Instead, we design an attack which leverages the fact that even though the source is not always the first node to communicate with curious nodes, with high probability it will be among the first to do so.
More precisely, the curious nodes record the 10 first nodes that communicate with them in each instance (results are not very sensitive to this  choice), and they predict the source to be the node that appears in the largest number of instances. In case of a tie, the curious nodes choose the node that first communicated with them, with ties broken at random. Figure~\ref{fig:composition_attack} shows that the precision of this attack increases dramatically with the number of rumors when $s$ is large, reaching almost sure detection for $10$ rumors. Remarkably, for small values of $s$, the attack precision increases much more gracefully with the number of rumors, as expected from the composition property of differential privacy discussed in Section~\ref{sec:related_priv}. Meaningful privacy guarantees can still be achieved as long as the source does not spread too many rumors.  

\section{Concluding Remarks}
\label{sec:conclusion}
This paper initiates the formal study of privacy in gossip protocols to determine to which extent the source of a gossip can be traceable. Essentially: (1) We propose a formal model of anonymity in gossip protocols based on an adaptation of differential privacy; (2) We establish tight bounds on the privacy of gossip protocols, highlighting their natural privacy guarantees; (3) We precisely capture the trade-off between privacy and speed, showing in particular that it is possible to design both fast and near-optimally private gossip protocols; (4) We experimentally evaluate the speed of our protocols as well as their robustness to source location attacks, validating the relevance of our formal differential privacy guarantees in scenarios where the adversary has some background knowledge.

Our work opens several interesting perspectives. In particular, it paves the way to the study of differential privacy for gossip protocols in \emph{general graphs}, which is  challenging and requires relaxations of differential privacy in order to obtain nontrivial guarantees. We refer to
 Appendix~\ref{app:general-graphs}
for a discussion of these questions.
Another avenue for future research is motivated by very recent work showing that hiding the source of a message can  amplify differential privacy guarantees for the \emph{content} of the message \cite{anonymity-dp,anonymity-dp2,blanket}.
Unfortunately, classic primitives to hide the source of messages such as mixnets can be difficult and costly to deploy. Showing that gossip protocols can \emph{naturally} amplify differential privacy for the message contents would make them very desirable for privacy-preserving distributed AI applications, such as those based on federated \cite{kairouz2019advances} and decentralized machine learning \cite{Bellet2018a}.

\bibliography{dp_gossip}

\newpage 

\appendix

\section{Model Extensions}
\label{app:remarks}

We kept our model relatively simple in the main paper to avoid unnecessary complexity in the notations and additional technicalities in the derivation and presentation of our results.
In this appendix, we briefly discuss various possible extensions. Basically, we make here the point that, although these generally lead to technical complications, they do not significantly impact privacy guarantees.

\subsection{Pull and Push-Pull Protocols}

Our study focus on the classic \emph{push} form of gossip protocols. This can be justified by the fact that, for regular graphs, synchronous push has asymptotic spreading time guarantees that are comparable with the push-pull variant \cite{giakkoupis2016asynchrony}. Besides, the differential privacy guarantees of any gossip protocol are limited by the probability that the first node informed by the source is a curious node, and we show this bound can be matched with push protocols.
Nevertheless, extensions of our results to pull and push-pull variants of gossip protocols \cite{karp2000randomized} are possible.
Forgetting mechanisms similar to the ones in Algorithm~\ref{algo:gen_sync_push} can be introduced for these protocols, i.e. nodes would have a probability $1 - s$ to stop disclosing information after each time they are pulled (if they do not pull someone with the information in between). Although slightly different, the optimal privacy guarantees would remain of the same order of magnitude. Yet, we expect pull guarantees to be even worse in the case $s=1$ because curious nodes could stop suspecting all nodes that they have pulled and that did not have the rumor. Besides, the pull protocol for $s=0$ would be even slower than its push counterpart.

\subsection{Eavesdropping Adversary}

Since we consider a complete graph, our formalization of the adversary as a fraction $f / n$ of curious nodes is closely related to an eavesdropping adversary who would observe each communication with probability $f / n$. Indeed, both models consider that each communication has a probability $f / n$ of being disclosed to the adversary. Most of our results are thus easily transferable to this alternative setting. The only difference would be that all nodes can be suspected in the eavesdropping model, thus introducing a $(1 - f/n)^{-1}$ factor each time we consider the population of non-curious nodes.

\subsection{Information Observed by the Adversary}
\label{app:information_observed}
We discuss three possible variants of the output observed by the adversary.

\subsubsection{Messages Sent by Curious Nodes} For simplicity of exposition, we considered that curious nodes only observe messages that are sent to them and not the messages that they send.
However, including the messages sent by curious nodes in their observed output would not impact the bounds on privacy (i.e., the guarantees for the algorithms). For the optimal algorithm, we only consider what happens during the first round, so including the messages sent by curious nodes does not change the result. This in particular implies that the fundamental limits of Theorem~\ref{thm:push_lb} remain the same (since the adversary observes strictly more information).
Similarly, for the parameterized algorithm, Theorem~\ref{thm:sync_push_s} is obtained by bounding the probability of a set $\hat{S}$. Then, we have $p(\hat{S}, S_{\rm out}) \leq p(\hat{S})$ where $S_{\rm out}$ is the sequence of messages sent by the curious nodes. In general, adding the messages sent by curious nodes to the output sequences has little or no impact on our results.

\subsubsection{Message Ordering} We assumed that the relative order of messages is preserved in the output sequence observed by curious nodes. This could be relaxed, as in real-world networks a message sent before another may well be received after it. One could for instance introduce a random swapping model to take this into account and investigate whether this weaker output leads to an improvement in the privacy guarantees. However, we argue that this improvement would be quite limited. 
First of all, it would not affect the privacy guarantees of the optimal protocol: since there is a single active node able to send a message at any given time, swapping is not possible. Therefore, the lower bound and the matching algorithm would not be affected by this change. Since parameterized gossip is almost privacy-optimal for small values of $s$ and swapping would only increase privacy, then we argue that the guarantees of parameterized gossip would be very similar in this case. Furthermore, even when several nodes are active at the same time (e.g., in Algorithm~\ref{algo:gen_sync_push} large $s$), the proofs can be adapted to work with counting the messages \emph{received} instead of the messages \emph{sent}. In this case, swapping is as likely to expose the source (making its messages arrive earlier) than to hide it (delaying the messages it sends).

\subsubsection{Global Timing} \label{app:strong_adversary}
In our model, we assume that curious nodes only have access to the relative ordering in which they received the messages but they have no information on the global time at which it was sent. This is justified in practice by the asynchrony and locality of the exchanges. We briefly discuss here how the privacy guarantees are affected if one considers a stronger adversary that has access to the number of times the \verb+tell_gossip+ procedure has been called.
Formally, this adversary observes the set $S = \left\{ (t, i, j) | (i,j) = (S_{\rm omni})_t, j \in \mathcal{C} \right\}$. This set can be turned into a sequence by ordering it by increasing values of $t$. Note that this is not a realistic adversary as gossip protocols naturally enforce partial observability of the events.

The following result quantifies the limits of privacy for this stronger adversary, which can be compared to the results of Theorem~\ref{thm:push_lb} in the main text. We can see that in the regime $\epsilon=0$ (total variation distance), the limits remain the same. However, achieving $\delta<f/n$ and prediction uncertainty is not possible against this stronger adversary. Note also that Algorithm 1 with $s=0$ remains optimal.

\begin{theorem}
If a gossip protocol satisfies $(\epsilon, \delta)$-differential privacy and $c$-prediction uncertainty then we have $\delta \geq \frac{f}{n}$ and $c=0$ in the strong adversary setting. Furthermore, these bounds are tight and matched by
Algorithm 1 when its parameter is set to $s=0$.
\end{theorem}

\begin{proof}
The fact that \verb+tell_gossip+ is called at least once and is first called on node $0$ still holds. Sequence $S^{(0)}$ now denotes the fact that node $0$ communicates with a curious node at time 0. Since the protocol is run on the complete graph, the node selected by \verb+tell_gossip+ is chosen uniformly within $\{0, ..., n - 1\}$, so a curious node is selected with probability $\frac{f}{n}$. We thus have $p_0(S^{(0)}) = \frac{f}{n}$. Besides, node $0$ cannot communicate with a curious node at time $0$ if node $1$ starts the rumor so $p_1(S^{(0)}) = 0$. For prediction uncertainty, using the same sequence $S^{(0)}$ yields $\frac{p_i(S^{(0)})}{p_0(S^{(0)})} = 0$ for all $i\neq 0$ and therefore $c=0$.

It remains to show that these bounds are matched by Algorithm~\ref{algo:gen_sync_push} with $s=0$. The fact that the only outputs that have a different probability if node $0$ starts (compared to the case when $1$ starts) are those in which $0$ (or $1$) communicates with a curious node for its first communication is still true with the stronger adversary. Then, we write $p_0(S_0=0)=p_1(S_0=1)=\frac{f}{n}$ and $p_0(S_0=1)=p_1(S_0=0)=0$. This ensures that $p_0(S^{(0)}) \leq p_1(S^{(0)}) + \frac{f}{n}$ (similarly for $S^{(1)}$), and the result follows.
\end{proof}

\subsection{Malicious Behavior}
\label{app:malicious}
We also assumed for simplicity that nodes are \emph{curious} but not \emph{malicious}, i.e., they follow the protocol. This is motivated by a practical scenario where a subset of nodes are simply being monitored by a curious entity.
If curious nodes can also act maliciously, they have three possible ways to affect the protocol (abstracting away the content of the information): emitting more, emitting less, or not choosing neighbors uniformly at random. If they emit more, they will inform more nodes, which makes it more difficult for them to locate the source. If they emit less (potentially not at all), then in the case $s < 1$, the protocol could stop before all nodes are informed. Yet, the privacy bounds are derived from the fact that the source forgets the information before communicating to a curious node. Choosing the neighbors they send the messages to reduces to the case in which they emit less (for they do not send messages to uninformed nodes) but without affecting protocol speed or termination (this does not reduce the number of active nodes). Thus, the impact on the observed output and therefore on the privacy would be minimal.
In the case $s=1$, malicious nodes have slightly more impact but remain quite small: this case only makes the set of informed nodes grow slightly slower. 

\subsection{Termination Criterion}

For simplicity, in all our gossip protocols we used a global termination criterion (the protocol terminates when all nodes are informed).
Termination without using global coordination is a problem in its own right that has been extensively studied (see for instance \cite{karp2000randomized}). Although some termination criteria could have a great impact on privacy, we argue that termination can be handled late in the execution so as to reveal very little about the beginning, hence avoiding any significant impact on privacy. For instance, it is possible to design a variant of Algorithm~\ref{algo:gen_sync_push} in which nodes only flip a coin with probability $s$ for a fixed number of times, and then stop emitting completely. This fixed number would have to depend on $s$, but then if it is large enough, it would guarantee both termination and privacy. Indeed, nodes would not communicate with curious nodes each time they are activated with high probability so this counter would actually provide very little information to the curious nodes. Determining how large this number of iterations should be, and the exact impact on privacy (which we argue is very small), is beyond the scope of this paper. 

\section{Delayed Start Gossip}
\label{app:weak_privacy_protocol}
Consider the protocol described in Remark~\ref{rem:naive}, which we call \emph{delayed start gossip}:
\begin{enumerate}
\item[1.] The source calls \verb+tell_gossip+ once to forward to an arbitrary node, say node $j$.
\item[2.] Node $j$ then starts a standard push protocol (Algorithm~\ref{algo:gen_sync_push} with $s=1$).
\end{enumerate}

This simple protocol is clearly optimal from the point of view of differential privacy in the regime $\epsilon=0$ (total variation distance). 
Indeed, if the first communication does not hit a curious node then the probability of a given output when two different nodes start the gossip is the same. It is also fast since it runs the standard gossip after the first round. 

Yet, this naive protocol has a major flaw. Indeed, when the first communication hits a curious node, the adversary can monitor whether the sender communicates with curious nodes again in the next rounds. If it does not, they can guess that the node is the source, and they will in fact make a correct guess with probability arbitrarily close to 1 for large enough graphs. On the other hand, when the sender communicates again with a curious node shortly after, they can be very confident that this node is not the source. Hence, it is possible to design a very simple attack with a very high precision (almost always right) and almost optimal recall (identifying the source with certainty every time the information is actually released, i.e. with probability $\frac{f}{n}$).

Making sure that the adversary is uncertain about its prediction is therefore a desirable property. This is captured by our notion of \emph{prediction uncertainty}. The following proposition formalizes the above claims.

\begin{proposition}
\label{prop:delayed}
We call $c_{ds}$ the prediction uncertainty constant of the delayed start protocol and we assume the ratio of curious nodes $f / n$ to be constant. Then $c_{ds} \rightarrow 0$ when $n \rightarrow \infty$.
\end{proposition}

More generally, it is in principle possible to prove similar results for any protocol in which the source node does not behave like other nodes. Indeed, if the special behaviour can be detected, then the adversary can know for sure the source of the rumor. This motivates the need for more involved protocols such as those covered by Algorithm~\ref{algo:gen_sync_push}.

\begin{proof}[Proof of Proposition~\ref{prop:delayed}]
The proof reuses some elements of the proof of Theorem~\ref{cor:impossibility_sync}. We consider the sequence $S_r^{(0)}$ such that node $0$ is the first node to communicate with a curious node ($S_0 = 0$) and then $r$ other nodes communicate with curious nodes before $0$ does ($S_i \neq 0$ for $i \in \{1, ..., r\}$). We denote by $t_0$ the time at which node $0$ gets the message and becomes active again (we refer here to the global order, although of course the curious nodes do not have access to it). Then, with the usual notations we have:
\begin{align*}
p_0\left(S_r^{(0)}\right) &= p_0(S_0=0)p_0\left(S_r^{(0)} | S_0 = 0\right)\\
&\geq \frac{f}{n} p_0\left(\cap_{i=1}^r S_i \neq 0 | S_0 = 0\right)\\
&\geq \frac{f}{n} p_0(t_0 \geq r)\\
&\geq \frac{f}{n} p_0 (n_c(r) \leq k^*) p_0(t_0 \geq r | n_c(r) \leq k^* ).
\end{align*}

Then, we recall from the proof of Theorem~\ref{cor:impossibility_sync} that
\begin{align*}
p_0 (n_c(r) \leq k) &= p\left(Binom(k, \frac{f}{n}) \geq r\right)\\
&= p\left(Binom(k, 1 - \frac{f}{n}) < k - r\right)\\
&= 1 - p\left(Binom(k, 1 - \frac{f}{n}) \geq k - r\right),
\end{align*} 
so if we set $k = \frac{2n}{f}r$ and use tail bounds on the binomial law (Theorem 1 of \cite{arratia1989tutorial}) then there exists a constant $H$ (only depending on $f / n$) such that $p_0 (n_c(r) \leq r \frac{2n}{f}) \geq 1 - e^{- r H}$. Therefore, we have:
\begin{equation}
\label{eq:c1rn}
p_0\left(S_r^{(0)}\right) \geq \frac{f}{n} \left(1 - e^{-r H} \right)  \left(1 - \frac{1}{n}\right)^{r \frac{2f}{n}} \geq C_1(r, n).
\end{equation}
The last line comes from calculations done in the proof of Theorem~\ref{cor:impossibility_sync}. We now study $p_1(S_r^{(0)})$. Since node $1$ started the protocol then it means that no other node (and in particular $0$) will stop emitting the message. Therefore, if node $0$ is the first to communicate with a curious node then it will remain active for the whole duration of the protocol. Consider that the first disclosure happens after $T_f$ communications. We can write:
\begin{align*}
p_1\left(S_r^{(0)}\right) \leq p_1(S_0=0) p_1\left(\cap_{i=1}^r S_i \neq 0 | S_0 = 0, T_f \leq t_f\right) + p_1(T_f > t_f).
\end{align*}

Since the fraction of curious nodes is constant, we can choose $t_f$ independently of $n$ or $r$ such that $p(T_f > t_f) \leq e^{-\frac{f}{n} t_f} \leq \frac{\epsilon}{4(n-f)}$ if $t_f = \frac{n}{f} \log\left(\frac{4(n-f)}{\epsilon}\right)$ in order to control the second term. Then, 
\begin{align*}
p_1\left(\cap_{i=1}^r S_i \neq 0 | S_0 = 0, T_f \leq t_f\right) \leq \prod_{t=t_f}^{t_f + r} \left(1 - \frac{f}{n} \frac{1}{t}\right)\leq e^{-\frac{f}{n} \sum_{t=t_f}^{t_f + r} \frac{1}{t}}.
\end{align*}

A series-integral comparison yields that if $r=\log^2(n)$ then $\exp\left(-\frac{f}{n} \sum_{t=t_f}^{t_f + r} \frac{1}{t}\right) \leq \frac{\epsilon}{4}$ for $n$ large enough. Finally, we use the fact that $p_1(S_0 = 0) \leq \frac{1}{n - f}$ to write that 
$p_1\left(S_r^{(0)}\right) \leq \frac{\epsilon}{2 (n - f)}$.

Finally, we observe that $C_1(\log^2 n, n) \rightarrow \frac{f}{n}$ when $n \rightarrow \infty$ where $C_1$ is defined in Equation~\ref{eq:c1rn}. In particular, $C_1(\log^2 n, n) \geq \frac{f}{2n}$ for $n$ large enough, so we have
\begin{equation}
\frac{p(I_0 \neq {0} | S_r^{(0)})}{p(I_0 = {0} | S_r^{(0)})} 
 = \sum_{i \notin \mathcal{C}\cup \{0\}} \frac{p_i( S_r^{(0)})}{p_0( S_r^{(0)})} \leq \frac{n}{f} \epsilon.
\end{equation}

Since $\epsilon$ can be picked arbitrarily small and $\frac{n}{f}$ is assumed to be constant then the previous ratio can be made arbitrary small.
\end{proof}

\section{Detailed Proofs}

\subsection{Privacy Guarantees}
\label{app:proofs_faster_push}

\begin{proof}[Proof of Theorem~\ref{cor:impossibility_sync}]
Intuitively, the proof relies on the fact that the event $\{t_d(0) \leq t_i(1)\}$ (node $0$ communicates with a curious node before node $1$ gets the message) becomes more and more likely as $n$ grows, hence preventing any meaningful differential privacy guarantee when $n$ is large enough. To formalize this, we study $S_r^{(0)} = \{S, S_t = 0 \text{ for some } t \leq r\}$, the set of output sequences such that the rank of node $0$ in the sequence is less than $r$. For a specific sequence to not be in $S_r^{(0)}$, there must have been at least $r$ communications (because $r$ nodes must have communicated with curious nodes), and none of them involved $0$ and a curious node. Therefore, if we note $n_c(r)$ the number of communications that actually happened before the output sequence reached size $r$, we have $n_c(r) \geq r$. Then, denoting by $C(t)$ the node that communicated with a curious node at time $t$ (with $C(t) = -1$ when the communication did not involve a curious node):
\begin{align*}
p_0( S_r^{(0)}) &= 1 - p\left(\cap_{t=0}^{n_c(r)} C(t) \neq 0\right) = 1 - \textstyle\prod_{t=0}^{n_c(r)} p\left(C(t) \neq 0\right)
\geq 1 - \prod_{t=0}^{r} \left( 1 - \frac{f}{n}\frac{1}{t + 1} \right),
\end{align*}
 where the last step comes from the fact that the probability of node $0$ to be selected at time $t$ is $\frac{1}{|I_t|} \geq \frac{1}{t}$ because at most one node is informed at each step and the active node is selected uniformly among informed nodes. We use the fact that $\log(1 + x) \leq x$ for any $x > -1$ on $x = - \frac{f}{n}\frac{1}{t + 1}$ to get:
\begin{equation}
\textstyle\prod_{t=0}^{r} \Big( 1 - \frac{f}{n}\frac{1}{t + 1} \Big) = \exp\Big(\sum_{t=0}^{r} \log\Big( 1 - \frac{f}{n}\frac{1}{t + 1}\Big)\Big) \leq \exp\Big(- \frac{f}{n}\sum_{t=0}^r \frac{1}{t + 1}\Big).
\end{equation}
Therefore, $p_0( S_r^{(0)})$ goes to $1$ as $r$ goes to infinity. We emphasize that we do not need to fix any network size for this result to hold since the ratio $f / n$ is assumed to be constant.

Then, for a given $r$ and for any $k >0$, $p(n_c(r) \leq k)$ is equal to  $p(Binom(k, \frac{f}{n}) \geq r)$ where $Binom(k, \frac{f}{n})$ is the binomial law of parameters $k$ and $\frac{f}{n}$. This is because it is the probability of having exactly $r$ successes with the sum of less than $k$ Bernoullis of parameter $\frac{f}{n}$, which is equal to the probability of having more than $r$ successes with the sum of $k$ Bernoullis of the same parameters. Therefore, $p(n_c(r) \leq k)$ is independent of $n$ and we can choose $k^*$ independently of $n$ such that $p(n_c(r) > k^*) \leq \frac{1}{n}$.
Then, we write that 
\begin{align*}
p_1( S_r^{(0)}) &= p_1( S_r^{(0)}, n_c(r)\leq k^*) + p_1( S_r^{(0)}, n_c(r) > k^*) \leq p_1( S_r^{(0)} | n_c(r) \leq k^*) + 1/n.
\end{align*}
This implies $p_1( S_r^{(0)} | n_c(r) \leq k^*) \leq p_1( 0 \in I_r | n_c(r) \leq k^*) \leq 1 - p_1(0 \notin I_r | n_c(r) \leq k^*)$. We know that only $r$ communications have reached curious nodes but the others have reached a random node in the graph, and there is at most $k^*$ of them, so finally:
\begin{equation*}
p_1( S_r^{(0)}) \leq 1 - \Big(1 - \frac{1}{n}\Big)^{k^*} + \frac{1}{n}.
\end{equation*}
We immediately see that $p_1( S_r^{(0)})$ goes to $0$ as $n$ grows because $k^*$ is independent of $n$, and we have shown above that $p_0( S_r^{(0)})$ goes to $1$ as $n$ grows. Since we must have that $p_0( S_r^{(0)}) \leq e^\epsilon p_1( S_r^{(0)}) + \delta$, we must have $\delta = 1$ if we want $\delta$ and $\epsilon$ to be independent of $n$.
\end{proof}

\begin{proof}[Proof of Theorem~\ref{thm:source_uncertainty_private_push}]
For any set of sequences $S \subset \mathcal{S}$ such that $p_0(S) > 0$:
\begin{align*}
\frac{p(I_0\neq 0 |S )}{p(I_0 = 0 | S)} = \sum_{i \notin \mathcal{C}\cup \{0\}} \frac{p_i( S )}{p_0(S)}\geq \sum_{i \notin \mathcal{C}\cup \{0\}} \frac{p_i(A_1 = \{ 0 \} ) p_i(S | A_1 = \{ 0 \})}{p_0(S)},
\end{align*}
where $A_1$ is the set of active nodes at round $1$. Because the state of the system (active nodes) is the same in both cases we can write that $p_i(S | A_1 = \{ 0 \}) = p_0(S)$. Besides, $p_i(A_1 = \{ 0 \} )$ corresponds to the probability that node $i$ sends a message to node $0$ and then stops emitting. Therefore:
$\frac{p( I_0 \neq 0 |S )}{p(I_0 = 0 | S)} \geq \left(1 - \frac{f + 1}{n}\right) (1 - s) > 0.$
\end{proof}

\subsection{Spreading time}
\label{app:speed}

\begin{proof}[Proof of Theorem~\ref{thm:speed_s}]

The idea of this proof is to rely on the ``determinism'' of gossip process, similarly to \cite{sanghavi2007gossiping}. This means that the gossip process very closely follows its mean dynamics. In our case, there is an added difficulty in the fact that extra randomness is introduced by the deactivation of the nodes. Yet, we precisely quantify the impact of this phenomenon on the results. We start by showing that if more than $k(s)$ nodes are informed at a given time, then with very high probability the number of informed nodes never drops below this fraction once it is reached. Therefore, a number of messages proportional to the size of the graph is sent at each round. The condition on $s$ for this to happen is written in Equation~\eqref{eq:c2_c1}. More formally, we fix $s \in (0, 1]$ and denote by $A_t$ the number of nodes that are active at round $t$, which is such that $A_t = \alpha_t n$. Then, we note 
\begin{equation}
\label{eq:f_dynamics}
f: \alpha \rightarrow 1 - p_{u}(\alpha) (1 - \alpha s),
\end{equation}
where $p_{u} (\alpha) = (1 - \frac{1}{n})^{\alpha n}$. Note that $f(\alpha) = \frac{1}{n} \mathbb{E}[A_{t+1} | A_t = \alpha n]$. To see this, we count the number of active nodes at time $t+1$. In total,  $A_t = \alpha n$ messages are sent at the beginning of the round. Therefore, for each node, the probability of having received a message at the end of the round is exactly $1 - p_u(\alpha)$ since each message has a $1/n$ probability to be sent to this specific node. In the end, $n \left(1 - p_u(\alpha)\right)$ nodes get the message in expectation. The rest of the active nodes at time $t+1$ is made of the nodes that were active, did not receive the message and did not deactivate, which represents a portion $n \alpha p_u(\alpha) s$ of the nodes. Then, one can see that the function $f$ is simply the sum of these 2 terms. We show by using that $\left(1 - x \right)^{y} \leq e^{-xy} \leq 1 - xy + \frac{x^2 y^2}{2}$ that for $\alpha \leq \alpha_s = \frac{s}{1 + 2s}$, we have: 
\begin{equation}
\label{eq:f_alpha_alpha}
f(\alpha) \geq \left(1 + \frac{s}{2}\right) \alpha.
\end{equation}
Then, we follow the same steps as in Lemma 15 in \cite{sanghavi2007gossiping}. We call $A_t$ the number of active nodes at round $t$, and $A_{t,m}$ the number of active nodes at round $t$ after $m$ messages have been sent (so during the round). Then, we can define $X_i = A_{t,i+1} - A_{t, i}$. $A_{t,i+1}$ only depends on $A_{t,i}$ and so does $X_i$:
\[
X_i = \bigg \{
  \begin{tabular}{lll}
  $1$  & with proba & $s ( 1 - \frac{|A_{t,i}|}{n})$ \\
  $- 1$ & with proba & $(1 - s) \frac{|A_{t,i}| - 1}{n}$\\
  $0$ & otherwise &
  \end{tabular}
\]
Then, we define the martingale $Z_i = \mathbb{E}[\sum_{i=1}^{A_t} X_i | X_1, \cdots, X_i, A_t]$. This allows us to write $A_{t+1} - n f(\alpha) = Z_0 - Z_{A_t}$. 
If we call $S_{k,t} = \sum_{i=k}^{A_t} X_i$ then for any $d \in \{ -1, 0, 1\}$:
\begin{align*}
\mathbb{E}&[S_{1,t} | X_1, , X_i, X_{i+1} = 1, A_t]\\&\geq \mathbb{E}[S_{1, t} | X_1, \cdots, X_i, X_{i+1} = d, A_t]\\
&\geq \mathbb{E}[S_{1, t} | X_1, \cdots, X_i, X_{i+1} = -1, A_t],
\end{align*}
 because the distribution of $X_i$ only depends on $A_{t, i}$. Therefore, $| Z_{i+1} - Z_i | \leq (1 +  \mathbb{E}[S_{i+1, t} | A_t + 1]) - (\mathbb{E}[S_{i+1, t} | A_t - 1] - 1)] \leq 2$. Azuma's inequality~\cite{mitzenmacher2005probability} then gives:
\begin{equation}
\label{eq:azuma_hoeff}
p\left(A_{t+1} - n f(\frac{A_t}{n}) \leq - \lambda A_t | A_t = k\right) \leq e^{- \frac{(\lambda k)^2}{8k}}.
\end{equation}
We also have that $p(A_{t+1} < k | A_t \geq k) \leq p(A_{t+1} \leq k | A_t = k)$. Then, for any $\alpha \leq \alpha_s$, Equation~\ref{eq:f_alpha_alpha} yields that for all $\lambda$: 
\begin{equation}\label{eq:azuma_easy}
    p\left(A_{t+1} \leq A_t\left(1 + \frac{s}{2} - \lambda\right) | A_t\right) \leq e^{-\frac{\lambda^2}{8}A_t}.
\end{equation}
We can then bound this expression by using Equation~\ref{eq:azuma_hoeff} with $\lambda = \frac{s}{2}$, leading to $$p(A_{t+1} < k | A_t \geq k) \leq e^{-\frac{s^2}{32}k} \hbox{ if } \alpha \leq \alpha_s.$$

Denoting by $N_{k,j}$ the number of messages sent between rounds $k$ and $j$, we can decompose over  $C \alpha^{-1} \log n$ rounds so that if $m$ is such that there are at least $\alpha$ active nodes at round $m$ then:
$$p(N_{m, m + C \alpha^{-1} \log n } \geq C n \log n ) \geq (1 - e^{- \frac{s^2 \alpha n}{32}})^{C \alpha^{-1} \log n },$$
because it is equal to the probability that the fraction of active nodes never goes below $\alpha$ for $C \alpha^{-1} \log n$ rounds. Therefore, if
\begin{equation}
\label{eq:c2_c1}
s^2 \geq \frac{32}{\alpha n} \log \frac{C \log n}{\alpha\log(1 - \delta)}, \hbox{ then } p(N_{m, m + C \alpha^{-1} \log n} \geq C n \log n ) \geq 1 - \delta.
\end{equation}

Equation~\ref{eq:c2_c1} gives a lower bound on the value of $\alpha$. Note that for a fixed $\alpha$, this lower bound goes to $0$ as $n$ grows so in particular, Equation~\ref{eq:c2_c1} is satisfied for $\alpha = \alpha_s$ if $n$ is large enough. It now remains to show that such a fraction $\alpha$ of active nodes can be reached in logarithmic time. Usual gossip analysis takes advantage of the exponential growth of the informed nodes during early rounds for which no collision occur. We have to adapt the analysis to the fact that nodes may stop communicating with some probability and split the analysis into two phases. 

In the rest of the proof, we prove that a constant fraction of the nodes (independent of $n$) can be reached with a logarithmic number of rounds. We first analyze how long it takes to go from $\mathcal{O} (\log n)$ to $\mathcal{O}(n)$ active nodes and then from $1$ to $\mathcal{O} (\log n)$.
Equation~\ref{eq:f_alpha_alpha} along with Equation~\ref{eq:azuma_hoeff} with $\lambda = \frac{s}{4}$ give that as long as $A_{t_0} (1 + \frac{s}{4})^{t} \leq \alpha_s n$ then
\begin{align*}
p\left(A_{t + t_0 +1} \geq A_{t_0} (1 + \frac{s}{4})^{t+1} | A_t = A_{t_0} (1 + \frac{s}{4})^{t} \right) \geq 1 - e^{- \frac{\alpha n s^2}{128}}
\end{align*}
for any $t \geq t_0$ such that $A_{t_0}\left( 1 + \frac{c}{2}\right)^t \leq n \alpha_s$ . Therefore, if we do this for all $t \leq t_{\alpha_s} = \frac{\log(\alpha_s n)}{\log(1 + \frac{s}{4})}$ rounds (so for a logarithmic number of rounds) then $
p\left(A_{t_{\alpha_s} + t_0} \geq n \alpha | A_{t_0} \right) \geq \big(1 - e^{- \frac{A_{t_0} s^2}{128}} \big) ^ {t_{\alpha_s}}$ because in this case, $A_t \geq A_{t_0}$ for $t \geq t_0$. Therefore, if 
\begin{equation}
\label{eq:c2_c2}
A_{t_0} \geq - \frac{128}{s^2} \log \left(1 - \left(1 - \delta\right)^\frac{1}{t_{\alpha_s}}\right), \hbox{ then } p(A_{t_{\alpha_s} + t_0} \geq n \alpha_s | A_{t_0}) \geq 1 - \delta.
\end{equation}
Using the fact that $\left(1 - x \right)^{y} \leq e^{-xy} \leq 1 - xy + \frac{x^2 y^2}{2}$ along with the fact that $\delta < 1 \leq t_{\alpha_s}$ to simplify Equation~\ref{eq:c2_c2}, we show that if $A_{t_0}$ satisfies:
\begin{equation} \label{eq:condition_A_t_0}
A_{t_0} \geq \frac{128}{s^2} \log\left( \frac{2t_{\alpha_s}}{\delta}\right) ,
\end{equation}
then it also satisfies Equation~\ref{eq:c2_c2}.

It only remains to prove that such an $A_{t_0}$ can be reached with $t_0$ logarithmic in $n$.
For this, use again Azuma inequality but on $\kappa$ consecutive rounds this time. Therefore, Equation~\ref{eq:azuma_easy} becomes, assuming that at least $A_t$ messages are sent at each round:
\begin{equation}\label{eq:azuma_easy_consecutive}
    p\left(A_{t+\kappa} \leq A_t\left(1 + \frac{\kappa s}{2} - \lambda\right) | A_t\right) \leq e^{-\frac{\lambda^2}{8 \kappa}A_t}.
\end{equation}
We apply this inequality for $\kappa = 2C\log(n) / s$ and $\lambda = C \log(n) / 2$ (which is valid because at least $A_0 = 1$ node is active at each round), which yields:
\begin{equation}
    p\left(A_{\kappa} \leq 1 + \frac{C \log(n)}{2} | A_t\right) \leq e^{- \frac{s \log(n)}{64}C}.
\end{equation}
In particular, for a fixed values of $C$, $s$, and $\delta$, then $p\left(A_{\kappa} \geq \frac{C \log(n)}{2} | A_t\right) \leq 1 - \delta$ for $n$ large enough. Finally, $t_{\alpha_s}$ is logarithmic in $n$ so similarly, Equation~\eqref{eq:condition_A_t_0} is satisfied for $t_0 = \kappa$ if $n$ is large enough. 

We conclude the proof by noting that
\begin{align*}
p&\left( N_{0, t_0 + t_{\alpha_s} + C \alpha^{-1}\log n} \geq C n \log n \right)\\
& \geq p\left(A_{t_0} \geq \frac{128}{s^2}\log\left(\frac{2t_{\alpha_s}}{\delta}\right)\right)
p\left(A_{t_{\alpha_s} + t_0} \geq n \alpha_s | A_{t_0} \geq \frac{128}{s^2}\log\left(\frac{2t_{\alpha_s}}{\delta}\right) \right)\\
& \qquad \qquad \times p\left(N_{t_{\alpha_s} + t_0, t_{\alpha_s} + t_0 + C \alpha^{-1}\log n} \geq C n \log n | A_{t_0 + t_{\alpha_s}} \geq n \alpha_s \right)\\
&\geq \left(1 - \delta \right)^3 \geq 1 - 3\delta.
\end{align*}

Finally, we have that  $t_0 \leq 2C \log(n) / s$, $t_{\alpha_s} \leq \log(n) / s$ and $1 / \alpha_s \leq 3 / s$ so in the end, $t_0 + t_{\alpha_s} + C \alpha^{-1}\log n \leq 6 C \log(n) / s$. Without loss of generality, $\delta$ can also be replaced by $\delta / 3$.
\end{proof} 

\begin{remark}[Extension to the Asynchronous Version]
The first part of the proof directly extends to the asynchronous algorithm by simply considering slices of time during which a set of $\alpha n$ nodes send $\alpha n$ messages, which essentially means constant time. Then, we consider a logarithmic number of slices. The phase from $1$ to $\mathcal{O}(\log n)$ active nodes requires sending a logarithmic number of messages and can thus be done in logarithmic time. Finally, phase 2 (going from $\mathcal{O}(\log n)$ to $\mathcal{O}(n)$ active nodes) consists in evaluating a logarithmic number of rounds during which a logarithmic number of nodes are active. Again, the only important thing is the number of messages sent (and not which node sent them) so using constant time intervals ensures that enough messages are sent between each pseudo-rounds with high probability. 
To summarize, it is possible to prove a statement very similar to that of Theorem~\ref{thm:speed_s} in the asynchronous setting, where the notion of rounds is replaced by constant time intervals. We omit the exact details of this alternative formulation.
\end{remark}

\subsection{Maximum Likelihood Estimation}
\label{app:proof_mle}

\begin{proof}[Proof of Theorem~\ref{thm:mle}]
We recall that $S$ is the output sequence observed by curious nodes, so that $S_0$, is the first node that communicates with a curious node. The source is noted $I_0$, as it is the first node informed of the rumor. The set $P$ is such that $p(I_0 = i) = 0$ if $i \notin P$. Recall that $t_c$ is such that $S_{t_c} \in P$ and $S_t \notin P$ for $0 \leq t < t_c$. By a slight abuse of notation, note $A_t$ the set of active nodes at the time where $S_t$ is disclosed (time of $t$-th communication with a curious node), so $S_t \in A_t$ for all $t$.

We know that for all $i \in P$ then $p((S_t)_{t < t_c}| A_{t_c}, I_0=i) = p((S_t)_{t < t_c}| A_{t_c})$ since $S_t \notin P$ for $t < t_c$. Similarly, $p((S_t)_{t \geq t_c}| A_{t_c}, I_0=i, (S_t)_{t < t_c}) = p((S_t)_{t \geq t_c}| A_{t_c})$ since the output after some time only depends on the active nodes at that time. Therefore, $p(S | A_{t_c}, I_0=i) = p(S| A_{t_c})$ for all $i \in P$, which critically relies on the fact that $t_c$ is the time of first disclosure of a node in $P$ (the first inequality would not hold otherwise). We note $\left[n\right] = \{1, ..., n\}$. We then write for $i \in P$:
\begin{align*}
p(I_0= i | S) &= \sum_{A \subset \left[n\right] }p(A_{t_c} = A | S)p(I_0=i | A_{t_c}=A, S)\\
&= \sum_{A \subset \left[n\right]}p(A_{t_c} = A | S)p(I_0=i | A_{t_c} = A)\\
&= \sum_{A \subset \left[n\right]: S_{t_c} \in A}p(A_{t_c} = A | S) \frac{p(I_0=i)}{p(A_{t_c} = A)} p(A_{t_c} = A | I_0=i).
\end{align*}

Let $j \in P \cap A_{t_c}$. If $i \in A_{t_c}$ then $p(A_{t_c}=A | I_0=i) = p(A_{t_c} = A | I_0=j)$. Otherwise, let us denote $E_{ij}(A) = \cap_{k \in A \backslash \{j\}}\{\text{k active at time $t_c$}\}\cap_{k \notin A \cup \{i\}}\{\text{k inactive at time $t_c$}\}$. This event represents the realization of $A_{t_c}$ for all nodes different from $i$ and $j$. We then write:
\begin{align*}
p&(A_{t_c}=A | I_0 = i) = p(\cap_{k \in A}\{k \in A_{t_c}\}\cap_{k \notin A}\{k \notin A_{t_c}\} | I_0 = i)\\
&= p(E_{ij}(A) | I_0=i) p(j \in A_{t_c}, i \notin A_{t_c} | I_0=i, E_{ij}(A)) \\
&= p(E_{ij}(A) | I_0=j) p(j \in A_{t_c}, i \notin A_{t_c} | I_0=i, E_{ij}(A)) \\
&\leq p(E_{ij}(A) | I_0=j) p(j \in A_{t_c}, i \notin A_{t_c} | I_0=j, E_{ij}(A)))\\
&= p(A_{t_c}=A | I_0 = j).
\end{align*}

The inequality comes from the fact that it is more likely that $j$ is active and $i$ is inactive if $j$ is the source than if $i$ is (i.e., if it is already the case at the beginning). This means that $p(A_{t_c} = A |I_0=i) \leq p(A_{t_c}=A | I_0=j)$ for all $i \in \left[n\right]$ and $j \in A_{t_c}$. Since the summation is over all $A$ such that $S_0 \in A$ (by definition of $S_{t_c}$ and $A_{t_c}$), and $p(A_{t_c} = A |I_0=i) \leq p(A_{t_c}=A | I_0=S_{t_c})$, we have for all considered $A$: 
\begin{align*}
p(I_0 = i | S) &=\sum_{A \subset \left[n\right]: S_{t_c} \in A}p(A_{t_c} = A | S)\frac{p(I_0=i)}{p(A_{t_c} = A)} p(A_{t_c} = A | I_0=i)\\
&\leq \sum_{A \subset \left[n\right]: S_{t_c} \in A}p(A_{t_c} = A | S)\frac{p(I_0=S_{t_c})}{p(A_{t_c} = A)} p(A_{t_c} = A | I_0=S_{t_c})\\
&= p(I_0 = S_{t_c} | S).
\end{align*}

This means that $S_{t_c}$ is more likely to be the source than any other suspected node when the adversary observes output $S$. Note that this requires uniform prior over nodes that can be suspected since we used the fact that $p(I_0 = i) = p(I_0 = S_{t_c})$ for all $i \in P$. For $i \notin P$, $p(I_0 = i | S) = 0 \leq p(I_0 = S_{t_c} | S)$.
\end{proof}

\section{Challenges of Private Gossip for General Graphs}
\label{app:general-graphs}

A natural extension of this work is to consider general graphs. We discuss in this section several aspects related to the natural privacy of gossip protocols in arbitrary graphs. In particular, we highlight the fact that problem-specific modeling choices are needed to go beyond the complete graph, and that even defining a notion of privacy that is suitable for all graphs is very challenging.  

\subsection{Average-Case versus Worst-Case Privacy}
Unlike the case of complete graphs, the location of curious nodes critically impacts the privacy guarantees in arbitrary graphs.
A naive way to deal with this issue is to randomize the location of curious nodes \emph{a posteriori}. Let us denote by $\mathcal{L}_{i,j}^f$ the set containing all subsets of nodes of size $f$ of the graph that do not contain $i$ and $j$. For fixed nodes $i$ and $j$, the set of curious nodes $\mathcal{C}$ is drawn from $U(\mathcal{L}_{i,j}^f)$, the uniform distribution over $\mathcal{L}_{i,j}^f$. For some parameters $\epsilon,\delta\geq 0$, privacy can be defined as follows: $\forall i,j \in \{0, ..., n-1\}, \  \forall S \in \mathcal{S}$ $$\mathbb{E}_{\mathcal{C} \sim U(\mathcal{L}_{i,j}^f)}[p_i(S, \mathcal{C}) - e^\epsilon p_j(S, \mathcal{C})] \leq \delta.$$
Note that $p_i(S, \mathcal{C}) = 0$ if the output sequence $S$ is not compatible with the set of curious nodes $\mathcal{C}$, i.e. if $(k,l) \in S$ and $k,l \notin \mathcal{C}$. To pick the curious nodes, it is possible to either pick a set of $f$ curious nodes at once or to pick each node (except for $i$ and $j$) with probability $f / n$. This randomized definition allows to prove a bound similar to that of Theorem~\ref{thm:push_lb} for arbitrary graphs. Indeed, the first node that receives the rumor has probability $\frac{f}{n}$ of being a curious node. However, such average-case notions of privacy are highly undesirable: in this case, no protection is provided against a (much more realistic) adversary that controls a fraction of nodes fixed in advance.

The worst-case approach consists in bounding the maximum difference instead of the expectation. This is the approach taken in our work for the complete graph (the max operator is implicit because the location of curious nodes does not matter in a complete graph). In the case of general graphs, the corresponding privacy definition is given by:
$\forall i,j \in \{0, ..., n-1\}, \  \forall S \in \mathcal{S}$,  $$\max_{\mathcal{C} \in \mathcal{L}_{i,j}^f}[p_i(S, \mathcal{C}) - e^\epsilon p_j(S, \mathcal{C})] \leq \delta.$$
We immediately observe that with this definition, it is impossible to have $\delta < 1$ as soon as there is a node in the graph with less than $f$ neighbors. This modeling choice is quite unrealistic as well because having a node surrounded by curious nodes means that the adversary actually believes this specific node has a strong probability of being the source and therefore put more sensors around it. A possible alternative would be to place curious nodes so as to bound the maximum privacy for any pair of nodes, and then evaluate the minimum privacy in this setting. This definition would mean that the adversary wants to be able to distinguish any pair of nodes as best as possible.

We see that choosing the locations of the curious nodes in an arbitrary graph is a complex problem that is heavily dependent on the topology of the graph and on the prior of the adversary on the locations of the curious nodes. Indeed, the adversary may simply want to isolate a sufficiently small group of nodes that have a high probability of being the source. 

\subsection{Relaxing the Differential Privacy Definition}
Differential privacy is a very strong notion that enforces indistinguishability between all pairs of nodes, in order to be robust to any prior information about who might be the source. In particular, an adversary should not be able to precisely identify the source even if it knows that only two nodes in the graph can be the source. Although it was possible to obtain meaningful privacy guarantees of this kind for the complete graph, this appears to be too strong of a requirement for some graph topology and location of curious nodes.
Consider for instance the extreme case of a line graph. It is clear that any non-trivial adversary can always distinguish between two segments of the line. This intuition directly extends to any graph which admits a cut with only curious nodes in it. 

A natural idea is to restrict the pairs of nodes that are required to be indistinguishable. Several ways of doing this may be considered. For instance, one could require that each node is indistinguishable from $k$ other nodes in the graph. Such relaxed definition could be obtained using the Pufferfish framework~\cite{kifer2014pufferfish}, which explicitly provides a notion of secret to protect. But how to choose such $k$ nodes based on the topology and how to characterize the optimal locations of curious nodes is very challenging.
Another direction could be to adapt the notions of metric-based differential privacy \cite{geoprivacy,metricprivacy} to design a notion of privacy where the required indistinguishability for a given node is a function of its distance to curious nodes in the graph, or to require that pairs of nodes become less indistinguishable with distance in the graph. Yet, it is not clear how to characterize the influence of the graph topology.

\subsection{Optimality of Algorithm~\ref{algo:gen_sync_push} with \hmath $s=0$}
We have seen in this section that the privacy guarantees for arbitrary graphs heavily rely on the particular privacy notion and that some recent privacy frameworks may provide tools to relax the classic differential privacy definition which is generally too strong for arbitrary graphs. We conjecture that for some of these relaxed definitions, the optimal algorithm for general graphs will be the same as in our case of the complete graph. Indeed, the strength of Algorithm~\ref{algo:gen_sync_push} with $s=0$ is to forget initial conditions quickly. In the complete graph, it does so in one step. In an arbitrary graph, the information about the part of the graph the source belongs to is still present after some steps, but the source should quickly be completely indistinguishable from its direct neighbors. In particular, attacks based on centrality \cite{shah2011rumors} are rather meaningless against this algorithm because spreading only occurs along a random walk in the graph. As in the case of the complete graph, Algorithm~\ref{algo:gen_sync_push} with $s>0$ is then likely to enjoy near-optimal privacy guarantees.   

\end{document}